\documentclass[mnsc,nonblindrev]{informs3}

\OneAndAHalfSpacedXI % current default line spacing
\usepackage{amssymb}
\usepackage{amsmath}
\usepackage{algorithm} 
\usepackage[noend]{algpseudocode} 
\usepackage{bm}
\usepackage{bold-extra}
\usepackage{anyfontsize}
\usepackage[dvipsnames]{xcolor}
\usepackage{refcount}
\usepackage{hyperref} %Always have this one 2nd to last
\usepackage{customlabel} %Always have this one dead last

\usepackage{natbib}
 \bibpunct[, ]{(}{)}{,}{a}{}{,}%
 \def\bibfont{\small}%

 \hypersetup{
  colorlinks    = true,
  citecolor     = blue,
  linkcolor     = blue,
  urlcolor      = blue,
  pdfinfo={ Title={Pricing Valid Cuts for Price-Match Equilibria},
            Author={Robert Day and Benjamin Lubin}}
}

\TheoremsNumberedThrough     % Preferred (Theorem 1, Lemma 1, Theorem 2)
\ECRepeatTheorems

\EquationsNumberedThrough    % Default: (1), (2), ...
\MANUSCRIPTNO{MS-XXXX-XXXX.XX}

\newcommand{\Description}{} %% Used by ACM templates, make noop here.

\def\ExampleHeaderFont{\bfseries\scshape}
\def\TheoremHeaderFont{\bfseries\scshape}

\makeatletter
\let\oldth@TH=\th@TH
\gdef\th@TH{%
  \oldth@TH
  \theorempreskipamount=12pt plus 3pt minus 2pt \theorempostskipamount=12pt plus 3pt minus 2pt
  \labelsep0.5em
  \def\@begintheorem##1##2{\normalfont\TheoremTextFont
        \item[\hskip\labelsep\normalfont\theorem@headerfont ##1\ ##2.]}%
  \def\@opargbegintheorem##1##2##3{\normalfont\TheoremTextFont
        \item[\hskip\labelsep\theorem@headerfont ##1\ ##2\ {\bf(##3)}.]}
}
\let\oldth@EX=\th@EX
\gdef\th@EX{%
  \oldth@EX
  \theorempreskipamount=12pt plus 3pt minus 2pt \theorempostskipamount=12pt plus 3pt minus 2pt
  \labelsep0.5em
  \def\@begintheorem##1##2{\normalfont\TheoremTextFont
        \item[\hskip\labelsep\normalfont\theorem@headerfont ##1\ ##2.]}%
  \def\@opargbegintheorem##1##2##3{\normalfont\TheoremTextFont
        \item[\hskip\labelsep\theorem@headerfont ##1\ ##2\ {\bf(##3)}.]}
}
\makeatother

\newcommand{\defsubexample}[1]{%
\newcounter{subexamplecounter\getrefnumber{#1}}%
}
\newcommand{\subexamplelabel}[2]{\customlabel{#2}{\ref*{#1}.\arabic{subexamplecounter\getrefnumber{#1}}}}
\newenvironment{subexample}[1]{%
\refstepcounter{subexamplecounter\getrefnumber{#1}}%
\medskip\noindent\normalfont\TheoremHeaderFont{Example~\ref*{#1}.\arabic{subexamplecounter\getrefnumber{#1}}}:}
{\medskip}

\makeatletter
\long\def\@makefigurecaption#1#2{{
    \EGT\FigureCaptionFontStyle \HD{9}{0}{\FigureNameFontStyle #1
      \kern16pt}\ignorespaces #2\HD{0}{0}}\endgraf}
\makeatother

\let\oldendproof\endproof
\renewcommand{\endproof}{\hfill\halmos\oldendproof}

\begin{document}
%%%%%%%%%%%%%%%%

%\setcounter{page}{1}
%\DOI{10.1287/xxxx.0000.0000}%

\RUNAUTHOR{Day and Lubin}

\RUNTITLE{Pricing Valid Cuts}

% \TITLE{Bundling Information Goods of Decreasing Value}
\TITLE{Pricing Valid Cuts for Price-Match Equilibria}

\ARTICLEAUTHORS{%
\AUTHOR{Robert Day}
\AFF{Department of Operations and Information Management, University of Connecticut,  Storrs, Connecticut 06269\\ 
	\EMAIL{robert.day@uconn.edu} \URL{}}

\AUTHOR{Benjamin Lubin}
\AFF{Department of Information Systems, Boston University,  Boston, MA 02215\\ 
	\EMAIL{blubin@bu.edu} \URL{}}

} % end of the block

\ABSTRACT{%
 We use valid inequalities (cuts) of the binary integer program for winner determination in a combinatorial auction (CA) as ``artificial items'' that can be interpreted intuitively and priced to generate Artificial Walrasian Equilibria.
We thus provide a method for converting a CA problem that admits only non-anonymous, nonlinear bundle prices into one that admits anonymous linear prices over the augmented item space, forestalling ex-post bidder complaints about opaque and strongly discriminatory pricing.
To this end, we introduce a refinement of the Walrasian equilibrium which we call a “price-match equilibrium” (PME) in which all prices are justified by providing an iso-revenue reallocation for the hypothetical removal of any single bidder.
We prove the existence of PME for any CA and characterize their economic properties and computation. We implement \emph{minimally artificial} PME rules and compare them with other prominent CA payment rules in the literature.   
}%

%\KEYWORDS{deterministic inventory theory; infinite linear programming duality;
%  existence of optimal policies; semi-Markov decision process; cyclic schedule}

%\HISTORY{This paper was first submitted on April 12, 1922 and has been with the authors for 83 years for 65 revisions.}

\maketitle
%%%%%%%%%%%%%%%%%%%%%%%%%%%%%%%%%%%%%%%%%%%%%%%%%%%%%%%%%%%%%%%%%%%%%%
% Paper body
%%%%%%%%%%%%%%%%%%%%%%%%%%%%%%%%%%%%%%%%%%%%%%%%%%%%%%%%%%%%%%%%%%%%%%

% Samples of sectioning (and labeling) in MNSC
% NOTE: (1) \section and \subsection do NOT end with a period
%       (2) \subsubsection and lower need end punctuation
%       (3) capitalization is as shown (title style).
%
%\section{Introduction.}\label{intro} %%1.
%\subsection{Duality and the Classical EOQ Problem.}\label{class-EOQ} %% 1.1.
%\subsection{Outline.}\label{outline1} %% 1.2.
%\subsubsection{Cyclic Schedules for the General Deterministic SMDP.}
%  \label{cyclic-schedules} %% 1.2.1
%\section{Problem Description.}\label{problemdescription} %% 2.

% Text of your paper here

\section{Introduction}
Walrasian Equilibrium (WE) prices support an efficient market outcome in which each bidder weakly prefers the bundle she is awarded by the market at the provided prices, and in which an item receives a positive price only if there is no excess supply of the item. A Combinatorial Auction (CA) is a market in which bidders express detailed preference information over combinations of indivisible items. In spite of the importance of WE as a foundation of microeconomic theory, it is well known that WE item prices may not exist in a general CA. 

\citet*{BM87} first showed in the CA context that a WE exists if and only if the linear programming relaxation of its \emph{winner determination problem} (WDP-LP) has an integer optimal solution. Here, we emphasize that their result holds in the context of linear anonymous prices on the \emph{original items} of interest, i.e., the items listed by the auctioneer for sale. Following the spirit of this approach, we too apply strong duality in the WDP-LP context to obtain WE prices, but we also generate valid inequalities or cuts to the WDP and then price those cuts to establish WE prices in any case where they would not otherwise exist. We call these \emph{Artificial Walrasian Equilibria} or AWEs. Philosophically, these cuts offer prices for \emph{artificial items} that the auctioneer ``didn't think to offer'' for sale, such as a license to own a package of a particular size, or for example, a license to take any pair of items from the set $\{A,B,C,D,E\}$. If the auctioneer had the foresight to set up and sell such licenses, then a normal WE would be available. Here, we let the auction generate these artificial items endogenously, sparing the auctioneer from the need to predetermine which combinatorial structure of artificial items would do the job.

Recall that the core of a game is the set of payoffs that cannot be \emph{blocked} by a coalition that could do better by trading among themselves. We show that there are many ways to meaningfully generate and price cuts as artificial items for a given CA, proving that any point in the core of the auction game can be supported as an AWE, and that a single core outcome may have several representations as an AWE. Given this large space of AWEs to explore, we introduce an important subset of WE prices which we call \emph{price-match equilibrium} (PME) prices, defined here as item prices that would still be supported if any one bidder were removed (without charging any other bidder more than her bid). Such prices are more \emph{justified} from each bidder's perspective than non-PME WE prices: each player observes that should they leave, the seller could still collect any price they refuse to pay.
%
%\textcolor{red}{While such an iso-revenue property naturally holds in the limit case of large-markets, we provide for its guarantee in the small-market settings where CAs are typically deployed.} \todo[inline]{do you like this?}

We show that PME prices always exist for a CA when artificial items are permitted. While not previously described in the literature, we emphasize the PME concept as a meaningful normative criterion for market design: any non-PME price outcome should be viewed as an ``unreasonable'' over-pricing by the seller, an artifact of applying the WE concept to general combinatorial markets na\"ively, and a practice that should be avoided. Each bidder should demand that the auction not charge her more than can be \emph{matched} by other bidders. In procurement auctions, such as those in public contracting, the direction reverses: working contractors should demand not to be underpaid simply because the auctioneer selected a non-PME linearization of the combinatorial market. Every bidder should demand that the amounts they pay (or are paid in procurement) are fully matched by a feasible set of competing bids. Through artificial-item pricing as presented here, this refined notion of price-matching (a stricter notion of fairness than WE) is available.

We show that minimum-revenue core-selecting (MRC) payments (as introduced in \citealt*{DR07}) can always be supported with PME prices, but other PME possibilities are also available. As a new alternative to MRC prices, we refine further to suggest the use of \emph{Minimally Artificial} PME payments, or MAP payments for short. These are selected such that the aggregate adjustment associated with the use of cuts is minimized, subject to PME being established, a way to stay ``as close to standard linear prices as possible.'' Generic MAP payments are then \emph{elaborated upon} to find a minimal set of specific supporting cuts (and thus specific linear prices) that uphold the MAP property.

Final MAP prices will sometimes correspond to existing payment schemes (like MRC, or VCG), but given other bids they may be distinct from any fixed pricing approach from the literature that we are aware of. The total-revenue characterization for a fixed set of bids is given by:
\begin{align}
\text{pay-as-bid}\geq\text{NWE}\geq\text{MAP}\geq\text{MRC}\geq\text{VCG} \label{REV}
\end{align}
where NWE represents any \emph{Natural} WE (i.e., without any cuts/priced artificial items) should they exist, and VCG is the classical \emph{Vickrey-Clarke-Groves} mechanism, known to be the unique dominant-strategy incentive-compatible mechanism under typical assumptions. We show through a series of examples that each inequality above can be strict in some CA instances and exact equality in other instances. When no NWE exists, we may still have non-PME AWEs with strictly greater revenue than MAP prices, or at the other extreme, for example when bidders have unit-demand, minimal NWE = VCG and the distinction among these four solution concepts disappears. We further show that unlike MRC payments, MAP prices do not suffer from the \emph{unrelated-goods problem} introduced by \citet{B13}, and that the commonly studied gross substitutes condition, while sufficient to imply NWE existence, does not imply that the NWE is a PME.

Overall, MAP prices emerge as an alternative pricing paradigm for both prescribing total payments and \emph{explaining} them, precisely where traditional WE fails and where MRC and VCG had previously offered only opaque lump-sum payments. In some cases, the latter two revenue inequalities may be equality, and so MRC and VCG may admit supporting MAP prices, which therefore offer an explainable breakdown into individual linear prices. In other cases, MAP prices will differ from VCG and MRC. 

\begin{figure}
    \centering
    \includegraphics[width=0.75\linewidth]{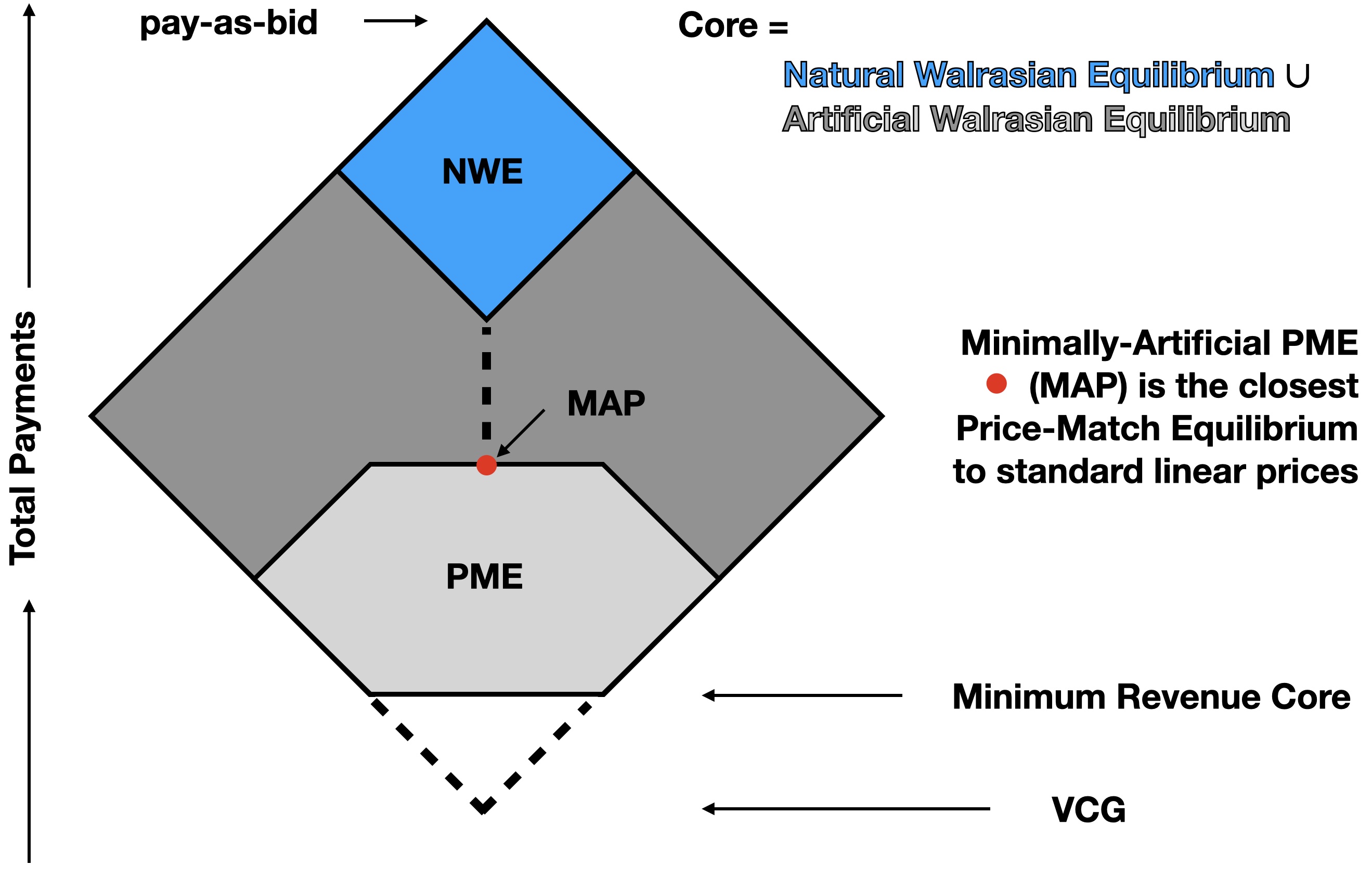}
    \caption{A generic diagram of the core in payment space with five alternative CA pricing rules, including our newly introduced MAP prices and the PME region of prices. The central dark gray region consists of non-PME points that can be supported by AWE prices.}
    \Description{}
    \label{fig:MAPcol}
\end{figure}
Figure \ref{fig:MAPcol} summarizes and contextualizes our contributions with a generic depiction of the core in payment space. The core itself is bounded above by bids (individual rationality constraints) and below by $\geq$-constraints that ensure subsets of winning bidders must weakly exceed the offers of potentially blocking coalitions. NWE are defined by envy-freeness constraints: the prices of the items I bid on must be large enough that I don't want to switch to a different package. The blue NWE region is known to shrink to non-existence, or expand to fill the entire core, including perhaps the VCG point, depending on the structure and strength of the bids. The region of PMEs, by contrast, are defined by $\leq$-constraints: payments must not be so much that they cannot be matched by losing bids. The PME region may contain outcomes with greater total payments than MRC (as depicted) or may shrink to  coincide with the MRC exactly, but PME always includes the MRC points, which themselves must always exist.

Pay-as-bid and the VCG payments represent the two extreme benchmark pricing rules for CAs in Figure \ref{fig:MAPcol}, outcomes that have been broadly studied and need little introduction. The latter are known to sometimes be outside the core, as depicted, a motivation for the use of MRC points which are by definition the smallest total-payment outcomes in the core. Together, VCG, MRC, NWE (when they exist), and pay-as-bid, give four potential payment outcomes for a given set of bids, all of which have been studied extensively and used in practice. Here, we introduce and explore the use of MAP prices as a fifth well-defined alternative outcome for any CA, an intermediate point which, as shown, often lies in the interior of the core.

To see one potential application of the results presented here, suppose MAP was used in the final (supplementary) round of the widely used combinatorial clock auction \citep{ACM2006}. If the market were to clear (with no excess supply) in the clock phase of such a CA, and if the associated NWE satisfies the newly defined PME property, a MAP auction would terminate at NWE with no change in payments between the clock phase and final outcome, given straightforward bidding. This would avoid the potentially hard-to-explain drop in revenue a seller could experience under a VCG or MRC final pricing rule, when agreeable NWE prices were already obtained in the clock phase. Example \ref{ex:Four} below shows an example of this type; clock prices would find NWE = MAP payments of 20 to clear the market, but a supplementary round with VCG or MRC would lower revenue to 10. Importantly, in other situations (like Example \ref{ex:Seven} below) all NWE may fail the PME condition, while a MAP price rule for the supplementary round would produce the more explainable PME outcome. Given the above revenue characterization, MAP prices offer an intermediate payment scheme for the market-design tool kit, between the traditional schemes in both revenue generation and truthfulness, which tend to be directly opposed as described by \citet*{ParkesKalagnanamEso2001}.

%\textcolor{BrickRed}{\textbf{Note:} The current draft is being circulated for feedback on the theoretical presentation, though we expect an additional computational study to provide a numerical evaluation using CATS and/or SATS data in a future draft. Please do not circulate without the authors' permission.}

\section{Background and Related Literature}
Real-world CA applications often use a combinatorial-clock format, with a first stage consisting of iterative ``clock'' rounds, price-vector announcements eliciting demand responses, followed by a sealed-bid CA final stage, as in \citet*{ACM2006}. Here, we focus on a the latter sealed-bid CA as a standalone process, which can be used as presented or as the final sealed-bid phase of a more complex mechanism.  For such a simple mechanism, bidders submit collections of bids and the auctioneer announces final winners and payments on behalf of the seller. Alternatively, this paper studies a \emph{solution concept} for the auction game, interesting in its own right as an explainable, agreeable outcome when all values are known.

We assume that we are working on problems of a practical size, ones for which we can comfortably solve the NP-hard WDP with available software, such as CPLEX or Gurobi. Under similar assumptions, in 2013 the UK used CAs to sell 28 licenses grouped into 4 types, then in 2014 Canada sold 98 licenses in 56 types, and then in 2015 318 licenses of 106 types, the latter aided by the use of a specialized $OR$-based bidding language \citep{AusBar17}. Thus, while a generic CA winner-determination problem is NP-hard and has poor asymptotic performance as problem size grows, many practical problem instances are far enough from such asymptotic behavior for the limitation to be of little concern. 

Walrasian Equilibrium (aka competitive equilibrium) is a central solution concept in microeconomics, with literature far too broad to be surveyed here. (See \citealt{BichFich21}, for a recent survey from the optimization perspective.) Among related work, \citet{baldwin2024strong} recently provided a strong-substitutes bidding language, a specific submission format (and thus preference restriction) for CAs in which Walrasian equilibrium prices (on the natural items) are guaranteed to exist. \citet{BichFux2018} apply WE principles to a real-world application in the regulation of fisheries, but in their case, when linear anonymous prices do not exist, they relax efficiency in favor of keeping linearity and anonymity of the prices (for natural items) strictly. Here, we instead explore artificial items to approximate WE, relaxing along a different dimension while strictly maintaining efficiency with respect to submitted bids. 

Approximation to competitive equilibrium is also used in allocation markets with fictional money by \citet{Bud2011}, a format that has proved useful in practice for seats in academic courses and other situations of heterogeneous preferences and high-demand. More recent work along this line is given by \citet{NV22} who investigate social-approximate equilibria when no WE exists, providing both approximation techniques via relaxation and preference restrictions that result in provable bounds on the relaxation error. Taken as a whole, this literature stream explores a number of ways to approach practical market design by approximating WE along various dimensions when they are not guaranteed to exist, as in the current work. 

When each bidder is interested in consuming at most one item (a.k.a., unit demand) the minimal NWE and VCG outcomes are guaranteed to be equivalent \citep{leonard1983elicitation}, and therefore the revenue comparison (\ref{REV}) above flattens, yielding strict equality among the minimal NWE, MAP, MRC, and VCG solution concepts. Under such a strict assumption on preferences, several known auction mechanisms (e.g., primal-dual auctions as in \citet{dSV07}, etc.) could be run to deliver the same outcome we prescribe here: VCG with a decomposition into linear-anonymous WE item prices. Here, we make no such preference assumptions and instead provide distinct MAP prices for cases in which unit demand and even less strict conditions like gross substitutes and coalitional submodularity \citep{dSV07} fail.

\citet{bikhchandani2002package} show that VCG is a WE if and only if a ``buyers are substitutes'' condition holds. They explored combinatorial exchanges (i.e., with potentially several sellers in addition to several buyers), and consider equilibria supported by potentially non-linear and non-anonymous prices. Under general preferences for CAs (i.e., with one seller under standard assumptions), they show the existence of instances for which only non-linear and non-anonymous pricing equilibria exist. Here, we show that if one is willing to accept an extended notion of linear prices on carefully selected \emph{artificial} items, then linear anonymous prices are possible whenever the core is non-empty, including all such CAs. 

%\todo[inline]{Modified below, see if you like the change...}

The introduction of artificial items results in two distinct notions of anonymous prices: (1) as in \citet{bikhchandani2002package}, prices are anonymous if any two bidders see the same total price for the same bundle of \emph{natural} items; (2) herein, prices are anonymous if any two bidders see the same total price for the same bundle of natural \emph{and artificial} items. Note that (2) does not imply (1); in Example~\ref{ex:Seven} (below) bidder 1 sees a price for bundle $\{B,C\}$ of 12, while bidder 2 sees a price of 8, violating (1). Yet, (2) holds because only bidder 1 would be required to buy artificial items priced at 4 along with $\{B,C\}$, based on the combinatorial structure, as will be explained below. Accordingly, interpreting artificial items as licenses the seller \emph{could have offered}, all bidders face the same traditionally anonymous price vector (i.e. according to definition (1) above) in this augmented setting, which is precisely what we mean by definition (2).
%
%That is, for the final prices we propose, while at times non-linear in natural items and non-anonymous under definition (1), the outcome is equivalent to a set of linear anonymous prices in an auction in which the artificial items were natural.
%
Mirroring the \citet{bikhchandani2002package} characterization of their non-linear non-anonymous equilibria as ``filling the core'' when there is only one seller, we show that the linear anonymous prices (under definition (2) above) that we introduce also fill the core in Theorem \ref{cuttocore}.

Also related to the current work are others that implicitly introduce cuts and price them. Among these, \citet{milgrom2022linear} find market outcomes that approximate natural WE when none exist. They do not make the distinction of desiring PMEs as introduced here, nor do they explore the concept of cuts as artificial items directly. Like \citet{oNeill2005efficient}, they have a single markup offset that can be interpreted in the language of the current paper as a single cut to the WDP-LP based on prices, similar to a tax on auction spending. The combinatorial cuts of the current paper attempt instead to generate artificial structures that are as \emph{simple} or \emph{explainable} as possible, in that we look for cuts reflecting simple relationships among bids that can be represented as integer knapsack constraints, with smaller constant terms favored as simpler and more explainable.

Even more closely related is the recent work of \citet{LahaieLubin2025} who introduce the adaptive (endogenous) generation of price terms in a dynamic (t\^atonnement) CA. They focus on CAs with unique items (supply of one item for each commodity) and ``monomial'' price terms that each correspond to a subset of items, allowing for the pricing of certain synergies directly, with the goal of terminating at WE prices. From the perspective of the current paper, each monomial price term can be viewed as a clique cut (i.e., a fictional item for which there is exactly one copy), though it can also be shown that not all clique cuts are implicitly considered. (For example, Example~\ref{ex:One}  below uses a clique cut that does not correspond to a unique subset or monomial term in that paper.) They are able to show computationally that these cuts are highly effective in generating WEs by augmenting the space of prices under consideration, complementing the perspectives offered here. Here, we also show a framework that adds pricing terms, but move to more general multi-unit CAs and explicitly search over a broader set of possible cuts. This richer set of cuts allows us to guarantee the stricter PME refinement of WE. 

Similarly, the concurrent work of \citet{ED24} uses the full class of clique cuts (but not cuts with an arbitrary constant term) to extend the price space in an attempt to better approximate WE prices when a NWE does not exist. The focus there is on larger problems and compact representation, rather than the details of finding AWE with desirable properties. Notably, the PME property is not considered in any of these related papers, which to the best of our knowledge is described for the first time here.

All AWE (and hence PME) prices naturally produce \emph{core-selecting} auctions, in which no coalition of bidders could propose a mutually (weakly) preferred outcome with higher total revenue. Core-selecting auctions are also described in \citet{DM08}, with the precise selection of core points discussed in \citet{DC12} and \citet{Bunz22}. We use the ideas of quadratic minimization as a tie-breaking rule as in the latter papers, but emphasize that here, quadratic tie-breaking is over item prices instead of personalized payments, an arguably better approach.  Here we also show that all core-selecting points can be reached as an AWE via just one valid cut, which allows us to easily select the core point we want: PME prices with a minimal amount of artificial price adjustment, then decomposed into a collection of more interpretable cuts.

\section{Model Preliminaries}

In a \emph{combinatorial auction} (CA), a single seller offers $J$ different types of indivisible items for sale to a set of $I$ bidders, denoted $\mathcal{I}=\{1,2,...i, ...I\}$.
The set of item types is denoted $\mathcal{J}=\{1,2,...j, ...J\}$, and each item $j$ has a supply of $c_j \geq 1$ identical copies available to be purchased. CAs with unique items, i.e., those with $c_j=1$ for all $j$, are often studied in the literature; here we allow for the more general case. Column vector $c$ will denote the complete supply vector with elements $c_j$.

Each bidder $i$ is interested in receiving a \emph{bundle} of items, represented as an integer column vector $a \in \mathbb{Z}^J_{\geq 0}$, with $0\leq a_j \leq c_j$ for each component $a_j$. 
Preferences are represented by a \emph{valuation function} $v_i:\mathbb{Z}_{\geq 0}^J\rightarrow \mathbb{R}_{\geq 0 }$, maintaining standard assumptions of normalization and monotonicity (aka free disposal), i.e., $v_i(\emptyset)=0$, and $a'\geq a \implies v_i(a') \geq v_i(a)$, for all $i \in I$ and any bundles $a'$ and $a$.

In a sealed-bid CA, the auctioneer collects a set of bids $\mathcal{K}=\{1,2,...k, ...K\}$ from the bidders, each one containing a bundle of interest $a^k$ and a monetary bid amount $b_k$. The subset of bids made by a particular bidder $i$ is $\mathcal{K}_i$, and when $|\mathcal{K}_i|=1$ we say bidder $i$ is \emph{single-minded}. Where convenient, we denote as $i(k)$ the unique bidder making bid $k$. While many ``bid languages'' have been considered in the literature for particular markets, general studies of CAs typically focus on this $XOR$ language, allowing for full expressivity. Stated simply, this means that the auctioneer may accept at most one bid from $\mathcal{K}_i$ for any bidder $i$. 

The bid amounts $b_k$ are collected into a row vector $b$, and the indexed set of all bundles (column vectors) associated with bids $a^k$ can be collected into a $J\times K$ matrix $\bm{A}$, i.e., each element $\bm{A}_{j,k}=a^k_j$ of $\bm{A}$ equals the desired quantity of item $j$ in bid $k$. As in other literature on Walrasian equilibria in CAs, we consider the existence and computation of prices as a solution concept while assuming bidders tell the truth. Thus, except where otherwise noted we assume truthful bidding, i.e., $b_i(a)=v_i(a)$ for all $i$ and $a$.

For small illustrative examples, we write $b_i(a)$ for bidder $i$'s bid amount on bundle $a$, or alternatively, for example, $b_3(ABD)=20$ to denote bidder $3$'s bid is \$20 for the set $\{A,B,D\}$, with an alternative naming of items by capital letters. We drop set braces in many places where the meaning should be clear, e.g., writing $\mathcal{I}\setminus i$ instead of $\mathcal{I}\setminus \{i\}$, and will say, for example, $b_1(9A)=85$ for a bid of 85 from bidder $1$ for $9$ copies of item $A$. 

All proofs not directly stated in the text are deferred to Appendix~\ref{appProofs}.

\subsection{Winner Determination}
To capture the $XOR$ nature of these bids, let $\bm{B}$ be the $I\times K$ matrix with elements $\bm{B}_{i,k}=1$ when $k \in \mathcal{K}_i$ and $\bm{B}_{i,k}=0$ otherwise. The standard winner determination problem (WDP) for the CA can be written as a binary linear program:
\begin{align}
    w(\mathcal{I}, \mathcal{J}, \mathcal{K})=\quad \max\quad & bx \label{WDP} \tag{WDP}\\
    \text{s.t.\quad} &\bm{A}x \leq c \label{suplfeas}\\
    &\bm{B}x \leq \bm{1} \label{xorfeas}\\
    &x \in \{0,1\}^K \label{integ}
\end{align}

\noindent where $\bm{1}$ is the vector of all ones of dimension $I$. This optimization selects a bid (column) $k$ by setting the corresponding decision variable $x_k$ to 1, with constraints (\ref{suplfeas}) ensuring that no more than $c_j$ copies of any item is sold, and (\ref{xorfeas}) upholding the $XOR$ language, i.e., at most one bid may be accepted for each bidder. A feasible solution to WDP is called a \emph{feasible allocation}. When bids are true values, i.e., $b=v$, an optimal solution is called \emph{efficient}. An optimal solution is said to be \emph{efficient with respect to bids $b$} where necessary to avoid confusion.

For a selected optimal solution $x^*$, bidder $i$'s \emph{allocated bundle} will be denoted $a^{i*}$, equal to the unique column of $\bm{A}$ with $x^*_k = 1$ and $k\in \mathcal{K}_i$ or the zero vector/empty bundle $\bm{0}$. (Note that some literature uses an $x$ or $X$ variable to denote an \emph{allocation} directly; here we use $x$ as a vector of binary \emph{selection} variables to choose each bidder's allocation from among fixed coefficient columns $a^k$, emphasizing both an integer programming (IP) convention and that in practice, with $XOR$, bidders submit bundles which are then treated as parameters.) The set of bidders receiving a nonempty bundle at $x^*$ will be denoted $\mathcal{I}^*$, with winning bids denoted $\mathcal{K}^*$.

The linear relaxation of WDP replaces (\ref{integ}) with $\bm{0} \leq x \leq \bm{1}$, and will be called WDP-LP. Because a solution to WDP-LP is potentially fractional, we write $w^f(\mathcal{I}, \mathcal{J}, \mathcal{K})$ for its objective function. Similarly, $x^f$ is a solution to WDP-LP, used when it is different from $x^*$. A \emph{valid inequality} or \emph{cut} in the current context is a new row added to WDP that is satisfied by all binary feasible solutions to WDP, but not necessarily all fractional feasible solutions to the relaxation WDP-LP. Because the final cuts we prescribe will have non-negative integer coefficients, the addition of any such cut results in a new WDP formulation equivalent to an extended CA with additional items but the same set of solutions, hence our interpretation of cuts as \emph{artificial} items.
 
\subsection{Payments, Prices, and Equilibrium}
In a sealed-bid CA, the auctioneer assigns an allocated bundle $a^{i*}$ and a monetary \emph{payment} amount $\rho_i$ to each bidder $i$. We assume \emph{quasilinear utility}: a bidder's payoff or utility is given by $u_i(a^{i*},\rho_i)=v_i(a^{i*})-\rho_i$. Such an auction mechanism is called \emph{individually rational (IR)} whenever $\rho_i \leq v_i(a^{i*})$ for all $i$, including $\rho_i=0$ when $a^{i*}=\bm{0}$. Similarly, the auction is \emph{IR with respect to bids} when $\rho_i \leq b_i(a^{i*})$ for all $i$. 

Of particular importance are \emph{linear anonymous} prices, where payments can be computed as $\rho_i=p a^{i*}$ for some row vector $p$ with linear anonymous prices $p_j \in \mathbb{R}_{\geq 0}$ for each item $j$. When a linear anonymous price vector $p$ is announced, each bidder can compute her \emph{demand correspondence}, defined as: 
\[
D_i(p)= \text{argmax}_{a \leq c} v_i(a)- pa
\]

\noindent This is the bidder's set of preferred bundles at the announced prices, with emphasis that this set could contain multiple bundles, including perhaps the empty bundle.

A solution to (WDP) and price vector $p$ are said to be \emph{envy-free} whenever each bidder prefers her allocated bundle at the current prices, i.e., when $a^{i*} \in D_i(p)$. An efficient, envy-free allocation-price pair are defined to be a \emph{Walrasian equilibrium} (WE) if also every item $j$ with excess supply has $p_j=0$. This latter condition can be represented as a complementary slackness condition in vector form: $p(c-\bm{A}x^*)=0$.

\begin{definition}
    An outcome is an \emph{artificial Walrasian equilibrium} (AWE) for a CA if and only if it is a WE for an augmented CA in which valid cuts of the original CA are treated like natural items with standard linear prices. 
\end{definition}

\subsection{Computing Dual Prices}

In this paper, we explore finding AWEs and will motivate and demonstrate how to select from among them over the next several pages. Anticipating the final outcome of this process (before we have fully specified it), assume that $\bm{A}$ contains a final set of added cuts that result in an AWE. It follows that WDP-LP will solve to integer optimality, and that dual prices form a WE via the main result of \citet*{BM87}. The dual of WDP-LP is:
\begin{align}
    \min\quad & pc + s\bm{1} \label{WDP-dual} \tag{WDP-dual}\\
    \text{s.t.\quad} & p\bm{A} + s\bm{B} \geq b \label{dual-env}\\
    & p \in \mathbb{R}^J_{\geq 0}\text{ ,}\quad s \in \mathbb{R}^I_{\geq 0} \label{dualnneg}
\end{align}

\noindent Each constraint (\ref{dual-env}) is of the form
$\sum_{j \in \mathcal{J}}a^k_j p_j + s_{i(k)} \geq b_k$, thus enforcing envy-freeness for non-winning bids. When WDP-LP is integral at optimality, $s_{i(k)}=b_k - \sum_{j \in \mathcal{J}}a^k_j p_j$ for each winning bid. Thus, dual variables $s_i$ are the \emph{surplus} amount for each bidder $i$, the \emph{apparent} utility, treating the bid as true value and subtracting all relevant prices in the bundle, with $s_i=u_i$ when $b=v$.

Solutions to WDP-dual are usually not unique, and often in the worst case for bidders, high prices giving $s_i=0$ for all $i$ are feasible. Thus, we will optimize over the set of solutions to instead ensure \emph{bidder-optimal payments}, as discussed in \citet*{AusMil02}, with the motivation of providing better incentives to bidders. While there may be a number of ways to select bidder-optimal payments, we provide one definitive selection process as:
\begin{align}
    \min \quad & p\bm{C}p^T \label{WDP-quad-dual} \tag{WDP-quad-dual}\\
    \text{s.t.\quad}  & (\ref{dual-env})-(\ref{dualnneg}) \notag \\
    &pc + s\bm{1} = b x^*\label{lockdual}
\end{align}

\noindent where (\ref{lockdual}) locks in the known optimal objective value to WDP-dual from a given solution $x^*$ to WDP, and where $\bm{C}$ is a $J\times J$ matrix with diagonal $\bm{C}_{j,j}=c_j$ and all other entries zero. This implements a quadratic minimization of squared prices (for $c_j$ copies of each item) over all solutions to WDP-dual.

\section{Pricing Valid Cuts as Artificial Items}
There are many methods for adding cuts to an IP whose LP relaxation is non-integral. Indeed there is a vast literature on the subject, generally dealing with the development of practical algorithms for solving IPs. We utilize a few of the tools from that literature stream (e.g., Chv\'{a}tal-Gomory cuts and lifting), but here, despite the NP-hardness of solving WDP in general (see, for example, \citealt{WD06}), we assume that we are working on problems of a practical size, ones for which we can comfortably solve the NP-hard WDP with available software, such as CPLEX or Gurobi.

\subsection{Motivation and examples}
Before presenting the details and properties of PME pricing, we show a few small examples to build intuition. We begin with a demonstration of the non-existence of WE prices for CAs.

\begin{example}
\label{ex:One}
$b_1(AB)= b_2(AC)= b_3(BC)=10;\quad \bm{b_4(ABC)=12}$
\end{example}

\noindent Here, we have $I=K=4$ single-minded bidders/bids, each bidding on exactly one bundle from a set of $J=3$ unique items, $\{A,B,C\}$. Bid $4$ is the only winning bid in the unique optimal solution to WDP. (We will indicate winning bids in these examples with $\bm{bold}$ font.) No IR prices constitute a WE, because envy-freeness demands:
\[
p_A+p_B \geq 10, \quad p_A+p_C \geq 10, \quad \text{and} \quad p_B+p_C \geq 10
\]
\noindent which combine to yield $p_A+p_B+p_C\geq 15$, a violation of IR with a payment of $15>12$, the value to the winning bidder. It is easy to show that the associated WDP has a valid cut:
\[
x_1+x_2+x_3+x_4 \leq 1
\]
\noindent Being a valid cut, adding this new row to $\bm{A}$ will not change the set of feasible solutions to WDP. Such an addition is equivalent to considering a similar auction given by $b_1(ABD)= b_2(ACD)= b_3(BCD)=10$; $b_4(ABCD)=12$, i.e., by adding a new item $D$ to every bundle. In this equivalent auction, WE prices are now given by $p=[0,0,0,10]$, which can be found using WDP-quad-dual.

Intuitively, the cut says, only one of these four bids can win. Thus, if we were to make a license $D$ and require that a bidder must buy $D$ to be eligible to win at least two items from $\{A,B,C\}$, the outcome would be equivalent to a second-price auction for $D$ and awarding $\{A,B,C\}$ each at zero (additional) cost. Though the auctioneer may have anticipated this and structured the items for sale to include such a license $D$, the point here is that she didn't have to; the auction format we propose finds and exploits the relevant structure to produce a WE, learning about the market and how to effectively price the items (including by possibly adding valid artificial items) based on the bids. 

One might also notice that the new WE outcome with $p_D=10$ is equivalent to the Vickrey-Clarke-Groves (VCG) auction outcome on the original bids. While our cut procedure does provide a price linearization of the VCG outcome on some problems like this one, this is not always the case, as we will see below with Example~\ref{ex:Five} and Corollary~\ref{VCG_PME}.

\begin{example}
\label{ex:Two}
$\bm{b_1(9A)= 85};\quad b_2(4A)= 40;\quad b_3(4A)= 40;\quad b_4(4A)= 35;$
\end{example}

\noindent Here, with a supply of $c_A=9$ units of a single type, and again single-minded bidders, $b_2(4A)=40$ dictates that $4p_A\geq 40$ and therefore $p_A\geq 10$. But this cannot be satisfied with IR constraint $9 p_A\leq 85$, so no (natural) WE is possible.

By multiplying the lone row of $\bm{A}$ by $0.25$ and rounding down non-integer coefficients (an instance of the Chv\'{a}tal-Gomory process), we find the cut $2x_1+x_2+x_3+x_4 \leq 2$, which can be priced $p=[0, 40]$, resulting in AWE. Here, the added cut can be thought of as a fictional item packaging $4$ items together, because the losing bids indicate that this is the minimum size of a valuable package. The winning bid on $9$ items therefore must buy two copies of this packaged quartet of items (reflected by the cut coefficient of $2$ for $x_1$) and receive one original (``unpackaged'') item free of charge, altogether paying $\rho_1=80$. 

\begin{example}
\label{ex:Three}
$\begin{aligned}[t]
&b_1(AB)= 18;\quad b_2(BC)= 17;\quad \bm{b_3(CD)= 18};\\
&b_4(DE)= 16;\quad \bm{b_5(AE)= 19};\quad b_6(ABCDE)= 34\\
\end{aligned}$
\end{example}

\noindent Again with single-minded bidders and unique items, here only two of the first five bundles could ever be accepted (considering only feasibility of item bundles), or, the sixth bundle could be accepted and knock out any other two bids. This statement is not difficult to read intuitively from the generated cut:
\[
x_1+x_2+x_3+x_4+x_5+2x_6 \leq 2
\]
\noindent Though no WE exists before the cut, with the cut we have a WE with $p=[2,0,1,0,0,16]$, resulting in $\rho_3=17$ and $\rho_5=18$.

Of interest, unlike the previous examples, this equilibrium uses non-zero prices on some original items in addition to pricing the generated artificial item. Items $D$ and $E$ are in the lesser valued losing bid $4$, which sets the value of the added cut item at $16$. This is the minimum price to take two items from $\{A,B,C,D,E\}$. With $p_B=0$ in a WE because $B$ does not sell, setting $p_A=2$ and $p_C=1$ covers the entire ``upsell'' above $16$ to satisfy envy-freeness on $b_1(AB)$ and $b_2(BC)$. Bid $6$ would face a price of $35$ (the price of $A$, $C$, and $2$ copies of the added item, per the added cut) which she does not envy. This cut demonstrates the ``odd-hole'' property discussed in the valid-cut literature.

\subsection{Price-Match Equilibria for CAs}

While it is possible to add just enough cuts to guarantee WE, we propose going further, with deeper cuts yielding what we call a \emph{Price-Match Equilibrium} (PME) which we now motivate. 

\citet{AG2005} define a \emph{winning level} of a bundle $a$ at any instant in a continuous efficient CA as the amount needed for a new bid on $a$ to just become winning, given the bids submitted up until that point in time. They show that it can be computed in our notation (before the new bid has been added to $\mathcal{K}$) as $w(\mathcal{I}, \mathcal{J}, \mathcal{K})-w(\mathcal{I}, \mathcal{J}\setminus a, \mathcal{K})$, where $\mathcal{J}\setminus a$ represents the same item classes but with supply vector reduced to $c-a$.

Each payment in Examples~\ref{ex:One}, \ref{ex:Two}, and \ref{ex:Three} is a winning level under this definition. For instance, in Example~\ref{ex:One}, remove the winning bid $k=4$ and find the level needed to just barely becoming winning (or to become strictly winning if just $\epsilon$ higher). Compute $w(\mathcal{I}\setminus 4, \mathcal{J}, \mathcal{K}\setminus 4)=10$ and $w(\mathcal{I}\setminus 4, \mathcal{J}\setminus ABC, \mathcal{K}\setminus 4)=0$, and thus $\rho_4=10-0$, consistent with the AWE computed for this example. Despite the prices having the winning-level property in these first examples, this is not true in general.

\begin{proposition}
\label{minWEgreater}
Minimal Walrasian Equilibrium prices may be strictly greater than winning-level payments. 
\end{proposition}

The proof is by example. We alter Example~\ref{ex:One} by raising the winning bid amount from $12$ to $16$, relaxing $b_4(ABC)$'s IR constraint:

\defsubexample{ex:One}
\begin{subexample}{ex:One}
\subexamplelabel{ex:One}{ex:OneDotOne}
$b_1(AB)= b_2(AC)= b_3(BC)=10;\quad \bm{b_4(ABC)=16}$
\end{subexample}

\noindent Now with a higher winning bid, WDP and WDP-LP no longer have an integrality gap. WDP-quad-dual prices $p=[5,5,5]$ give a WE with $\rho_4=15$ for Example~\ref{ex:OneDotOne} (compared to $\rho_4=10$ in Example~\ref{ex:One}, after adding a cut). Note that $p$ is minimal in the available payment space, i.e, with no cut and pricing only three items. This provides strong motivation to add a cut, even when there is no integrality gap; the seller has no feasible way to generate revenue of more than $10$ without bidder $4$, so why should bidder $4$ be forced to pay 50\% more? Applying the same cut as in Example~\ref{ex:One} to Example~\ref{ex:OneDotOne}, results in the same $\rho_4=10$, which again is the VCG outcome.

Indeed, in Example~\ref{ex:One}, without the cut, the first three bids would be ``half accepted'' by solution $x^f=[0.5, 0.5, 0.5, 0]^T$ to WDP-LP. In doing so, it's as if the three bids are allowed to represent $15$ units of value, when really they only generate $10$ units of value integrally. In Example~\ref{ex:OneDotOne} (if a cut is not applied), the raised bid ``hides'' this non-convex over-influence of these three bids, allowing them to drive prices up beyond their ability to match those prices in any integral solution. 

To further refine the winning-level notion for the current context (i.e., a multi-unit, sealed-bid CA with linear prices and multi-minded bidders), consider the following definition. As is typical, we use a $-i$ superscript to denote a situation in which bidder $i$ is removed or zeroed out while all other bidders remain.

\begin{definition}[Price Match Equilibrium] \label{def:PME}
A WE allocation $x^*$ and price vector $p$ constitute a \emph{Price-Match Equilibrium} (PME) if and only if $\forall i$, there exists an allocation $x^{-i}\in \{0,1\}^K$ with $x^{-i}_k=0$ for $k \in \mathcal{K}_i$ and for which the seller could generate the same total payments (i.e., for which $p\bm{A} x^{-i}= p\bm{A} x^*$) subject to IR with respect to bids.
\end{definition}

Next, let $K(p)$ represent the original bids but with amounts $b$ changed to $b^p$ where $b^p_k=pa^k$ if $pa^k \leq b_k$, $b^p_k=0$ otherwise, i.e., $K(p)$ replaces each bid amount with zero or the total linear price for the associated bundle $a^k$ according to $p$ whenever that price would yield non-negative surplus. We then have the following:

\begin{proposition}
\label{prop:winninglevel}
WE allocation-price pair $(x^*,p)$ is a PME if and only if:
\[
\rho_i=pa^{i*}=w(\mathcal{I}\setminus i, \mathcal{J}, \mathcal{K}(p))-w(\mathcal{I}\setminus i, \mathcal{J}\setminus a^{i*}, \mathcal{K}(p)) \qquad \forall i \in \mathcal{I}
\]
\end{proposition}

According to this result, in a PME each winning bidder's payment $\rho_i$ is (simultaneously) in the functional form of a winning level; each $i$ sees prices that the other bidders are willing to pay, both with and without her own consumption of $a^{i*}$, and thus that her own payment if lowered would not be winning. For any CA with IR, the condition is trivially satisfied for all non-winning $i$ paying $\rho_i=0$. The WE outcomes of Examples~\ref{ex:One}, \ref{ex:Two}, and \ref{ex:Three} can easily be shown to be PMEs, while Example \ref{ex:OneDotOne} with $p=[5,5,5]$ is a WE but not a PME. 

Directly from the definition of PME, the second-price auction for a single item has a PME outcome. Similarly, the Generalized Second-Price Auction, in which the winner of the $n$th slot pays the amount of the $n+1$st bid (famous for its use in Google's ad-word auctions), prescribes a PME outcome: if any bidder were to leave, each bidder in a worse-ranked slot could be moved up one slot without violating IR.  \citet{leonard1983elicitation} shows that when unit-demand holds, VCG is the unique minimal NWE. Together with Corollary \ref{VCG_PME} below, we have:

\begin{proposition}
    Unit demand is sufficient to guarantee min-NWE = PME = MRC = VCG.
\end{proposition}
\noindent where min-NWE is the payment-minimal NWE, which is unique in that setting.

In a WE, the seller prefers to sell the efficient allocation, here represented as: $p\bm{A}x^*=w(\mathcal{I}, \mathcal{J}, \mathcal{K}(p))$. From WDP's definition, removing a bidder yields $\sum_{\bar{\imath}\in \mathcal{I}\setminus i}\rho_{\bar{\imath}}=w(\mathcal{I}\setminus i, \mathcal{J}\setminus a^{i*}, \mathcal{K}(p))$, which added to the condition of Proposition~\ref{prop:winninglevel} yields an alternate characterization of PME:
\begin{proposition} \label{prop:noloss}
WE allocation-price pair $(x^*,p)$ is a PME if and only if:
\[
p\bm{A}x^*=w(\mathcal{I}\setminus i, \mathcal{J}, \mathcal{K}(p)) \qquad \forall i \in \mathcal{I}
\]
\end{proposition}
PMEs are a (sometimes strict) subset of the AWEs associated with a CA. Examples below show that for the same CA instance, several PMEs may be available, with some relying more heavily on artificial item prices than others, leading us to the following two concepts.  

\begin{definition}
    The \emph{total artificiality} of an AWE is $\sum_{\mathcal{J}^+}\sum_{\mathcal{K}}p_ja_{jk}$, where $\mathcal{J}^+$ is the set of all added artificial items.
\end{definition}
\begin{definition}
    A \emph{minimally artificial PME} (MAP) is a PME for which total artificiality is minimal among all PMEs.
\end{definition}
Note that total artificiality includes all bids of all bidders (not just winning bids) in its summation. MAP outcomes will be shown to always exist for a CA and may not be unique. Directly from the definition, when a NWE exists, all NWE prices are MAP outcomes with zero total artificiality.
\section{Comparison to Min-Revenue Core-Selecting and VCG Mechanisms}
Using our notation, truth-inducing VCG payments are computed as:
\begin{equation}
\label{eq:VCG}
\rho^{VCG}_i=w(\mathcal{I}\setminus i, \mathcal{J}, \mathcal{K})-w(\mathcal{I}\setminus i, \mathcal{J}\setminus a^{i*}, \mathcal{K})
\end{equation}
\noindent Minimum-Revenue Core-Selecting (MRC) auctions minimize $\sum_{i\in \mathcal{I}}\rho_i$ such that:
\begin{equation}
\label{eq:coreConst}
\sum_{i \in \mathcal{C}}(b_{i*}-\rho_i) +  \sum_{i \in \mathcal{I}}\rho_i\geq w(\mathcal{C}, \mathcal{J}, \mathcal{K^{\mathcal{C}}}) \quad\forall \mathcal{C} \subseteq \mathcal{I}
\end{equation}
\noindent where $\mathcal{K^{\mathcal{C}}}$ is shorthand for $\cup_{i \in \mathcal{C}} \mathcal{K}_i$.

\begin{proposition}\label{WEcore}
    Any WE pair $(x^*,p)$ is core-selecting.
\end{proposition}

The proof sketch is that core-selecting constraints as in Equation~\ref{eq:coreConst} can be derived by simply adding the relevant envy-freeness constraints (\ref{dual-env}) indexed by winning bids in the $w(\mathcal{C}, \mathcal{J}, \mathcal{K^{\mathcal{C}}})$ outcome. (This result is well known in various forms. See, e.g., \citealt{bikhchandani2002package}.) The converse does not directly hold: not all core outcomes can be reached via the currently available items prices. (This will be addressed later by Theorem \ref{cuttocore}, however.)  

To demonstrate this important difference between MAP prices and the VCG and MRC benchmarks from the literature, consider the following numerical example, adapted from \citet{B13}:

\begin{example}
\label{ex:Four}
$\begin{aligned}[t]
&\bm{b_1(AC)= 21}; \quad b_1(A)= 10; \quad b_1(C)= 11;\\
&b_2(BC)= 20;\quad \bm{b_2(B)= 10}; \quad b_2(C)= 10;\\
&b_3(AB)= 10\\
\end{aligned}$
\end{example}

% \begin{figure}
%     \centering
%     \includegraphics[width=0.9\linewidth]{Example4.png}
%     \caption{The polyhedron of natural Walrasian equilibria (in price space) for Example 4, depicted over the core (projecting payment space onto the xy-plane). The PME$^*$ and pay-as-bid points in both regions are labeled, with the MRC $=$ VCG point labeled in the core. The image of NWEs is unshaded in the core while the AWE $=$ PME points are shaded. The black line segment intersection of these two regions is the set of MAP core points. WDP-quad-dual selects PME$^*$ (red point) from among the PME$^*$ segment.}
%     \label{fig:BO}
% \end{figure}
\begin{figure}
    \centering
    \includegraphics[width=0.5\linewidth]{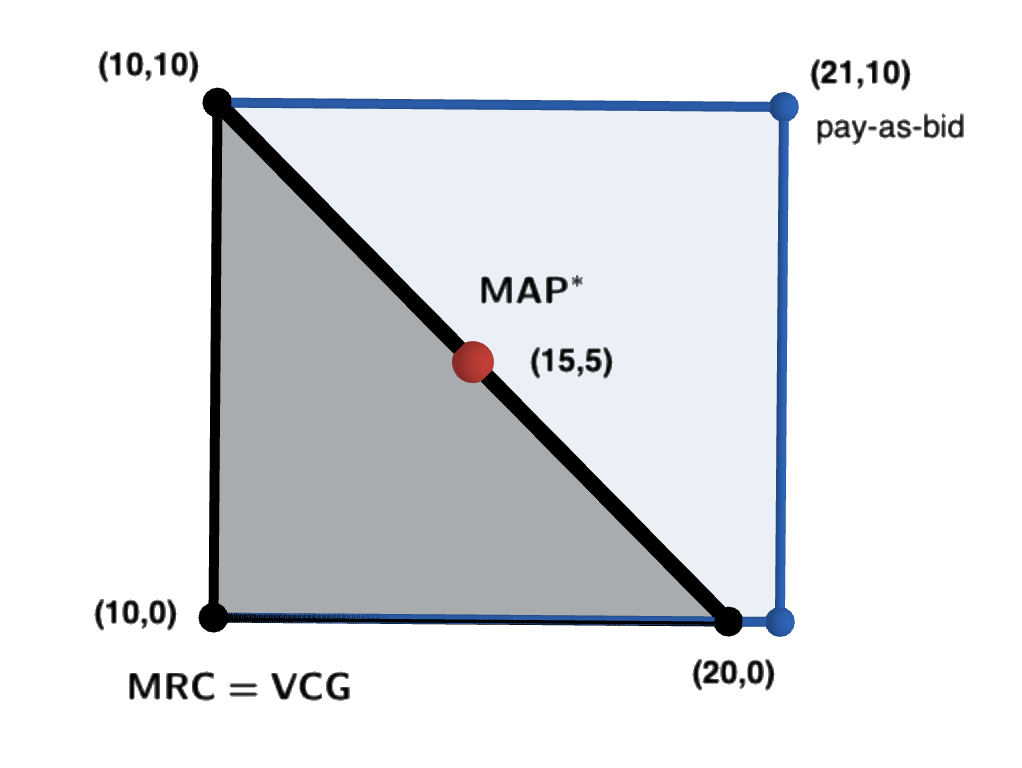}
    \caption{The core of Example~\ref{ex:Four} in the $\bm{\rho_1}$,$\bm{\rho_2}$ plane. The lighter region consists of NWE payment outcomes while the AWE $\bm{=}$ PME outcomes form the gray shaded triangle. The thicker line segment intersection of these two regions is the image of MAP prices on the core. WDP-quad-dual selects MAP$^*$ (red point) from among the PME points. The pay-as-bid and MRC $\bm{=}$ VCG points are also labeled.}
    \Description{}
    \label{fig:BO}
\end{figure}
\noindent \citet{B13} demonstrates a weakness of MRC, referred to as the ``unrelated goods problem,'' motivated by the observation that bidders $1$ and $2$ have preferences for $C$ that simply add to their values for either $A$ or $B$, respectively. An auctioneer considering holding separate MRC auctions for $\{A,B\}$ and $\{C\}$, or selling $\{A,B,C\}$ together would determine:
\begin{align*}
&\rho^{MRC}_1 = 5\text{ for A}&&\rho^{MRC}_2 = 5\text{ for B} &&&\text{when selling } \{A,B\}\\
&\rho^{MRC}_1 = 10\text{ for C} && &&&\text{when selling } \{C\}\\
&\rho^{MRC}_1 = 10 \text{ for AC} &&\rho^{MRC}_2 = 0 \text{ for B} &&&\text{when selling } \{A,B,C\}
\end{align*}
\noindent and observe that total revenue goes down drastically (from $20$ to $10$) when combining the $\{A,B\}$ auction with that of an ``unrelated good'' $\{C\}$ into a single auction. \citet{B13} notes that this shows a similar ``low-revenue'' problem as VCG, observing that total payments collected by VCG are $10$ whether these two auctions are conducted separately or together. With this example as motivation, we observe the unrelated goods problem vanishes for MAP prices as proposed here.

Indeed, for Example~\ref{ex:Four}, we have WDP-quad-dual prices $p=[5,5,10]$ forming a PME without the addition of any artificial cut items. Further, we have the same prices supporting two auctions if $\{A,B\}$ and $\{C\}$ are conducted separately, with total revenue of $20$ for the seller in either case. One might also observe that any $p$ with $p_A+p_B=10$ and $p_c=10$ are sufficient for PME here, or that any $\rho^{MRC}_1+\rho^{MRC}_2=10$ satisfy MRC when selling just $\{A,B\}$; this motivates our use of quadratic rules to specify a unique pricing, to select exactly MAP$^*$, as shown in Figure \ref{fig:BO}.

Intuitively, because MRC only demands personalized payments, in Example~\ref{ex:Four} bidders $2$ and $3$ cannot ``tell the difference'' between payments for $A$ and payments for $C$, and thus a form of ``double counting'' occurs. Bidder $2$ sees bidder $1$'s MRC payment of $10$ and would not (in a hypothetical sense) attempt to block the outcome and pay the seller more, as if bidder $1$'s payment was for $C$. But at the same time, bidder $3$ sees bidder $1$'s payment and also cannot block, as if the payment of $10$ was for $A$. By accounting for the prices of items separately, MAP prices avoid this double counting, and thus remain greater than VCG = MRC for this example.

\section{Generating Cuts for WE and PME}
We first consider generating a single valid cut. Let its coefficients be $\alpha_k$ for each $k \in \mathcal{K}$, collectively the vector $\alpha$, and its right-hand side be $\alpha_0$:
\[
\alpha x \leq \alpha_0
\]
which must hold for any WDP-feasible allocation $x\in \{0,1\}^K$. These $\alpha$ terms will be variables during the cut generation process and then later treated as constants in the final price determination, functioning exactly as all other items in $\mathcal{J}$, and concatenated as a new row of the WDP matrix $\bm{A}$.
This new cut will thus be treated as an \emph{artificial item} in the auction. Eventually, we prefer a collection of artificial items with each $(\alpha,\alpha_0)\in \mathbb{Z}^{K+1}_{\geq0}$ for interpretability, but for intermediate steps in our reasoning and algorithms we will also consider ``unscaled'' or continuous valid cuts, which we denote with $(\phi,\phi_0)\in \mathbb{R}^{K+1}_{\geq0}$.
%and thus in future iterations it will be converted to a new row in the bundle-bid-defining matrix $\bm{A}$ and treated like every other auction item. At any point in the process we assume updated auction parameters, so that $\bm{A}$ may contain one or more rows that were artificially generated in previous iterations, consisting of what were previously $\alpha_k$ variables that are now $a^k_j$ constants, with $\mathcal{J}$ increased in size accordingly. 
  
Any newly added artificial item must be \emph{fully demanded} at our fixed WDP solution $x^*$, i.e., $\alpha x^*=\alpha_0$ or $\phi x^*=\phi_0$. That is every copy of an artificial item must be purchased by the winning bids. If not, WE fails by definition, and pricing an excess-supply item may violate the core property, even when envy-freeness constraints (\ref{dual-env}) are satisfied. Because the preponderance of previous work on cut generation was used to \emph{solve} IPs like WDP, this assumption that we \emph{know} $x^*$ and constrain that $\alpha x^* = \alpha_0$ when looking for cuts is peculiar to the current line of research, to our knowledge.

%If WDP-LP solves to integral optimality without the addition of cuts, we refer to the result as a \emph{Natural} Walrasian Equilibrium (NWE), otherwise, a WE arrived at after the addition of artificial items is an \emph{Artificial Walrasian Equilibrium} (AWE). For formal statements below WE includes both AWE and NWE.

\subsection{Cuts for WE} \label{cutsforWE}

Though it might be computationally expensive in practice, one can in theory enumerate all set-wise maximal feasible solutions to WDP, those for which no other bid could be feasibly added. Let all such maximal feasible solutions be indexed by $\ell \in \mathcal{L}$ with elements $x^\ell$. If we construct the matrix of directions to optimality, $\bm{D} = [x^* - x^1, x^*-x^2,..., x^*-x^{|\mathcal{L}|}]$, then one direct approach to find a cut $(\alpha,\alpha_0)$ is by solving:
\begin{align}
    \max\quad & \alpha x^f \label{cutWE-obj} \tag{cut-frac}\\
    &\alpha \bm{D} \geq \bm{0} \label{cutWEvalid}\\
    &\alpha x^* = \alpha_0 \label{cutWEtight}
\end{align}
where $x^f$ is a fixed fractional solution to WDP-LP, treated here as a constant, as is $\bm{D}$. (If WDP-LP admits an integral solution, we are at a WE, and so we would terminate without the need for a new cut.) While constraints (\ref{cutWEvalid}) ensure that $(\alpha,\alpha_0)$ is a valid cut for WDP, any positive scalar multiple is also valid, resulting in an unbounded objective. To avoid this unboundedness, one can begin with $\alpha_0=1$ and later scale ($\alpha, \alpha_0$) to find an integer cut, as described in Algorithm~\ref{alg:master}, shown in Appendix~\ref{appA1}.
In practice, we do not need to construct the full $D$ matrix in order to solve this problem; we provide details on its dynamic generation in Appendix~\ref{DynD}.
%\todo[inline]{consider reordering Appendices so the lettering is alphabetical as introduced}

\subsection{Cuts for any Core Point and PME}

Example~\ref{ex:OneDotOne} demonstrated that a WE may not be a PME, motivating us to consider adding additional cuts to the WDP to get to PME. Algorithm~\ref{alg:master} (see Appendix~\ref{appA1}) can take us no further in generating new cuts once we have an AWE (or if we start at NWE), as its progress was driven by closing an IP gap that is now closed. For another, perhaps more interesting illustration, consider the single-minded bidders in the Example~\ref{ex:Five} that follows.  Here there are $c_A=17$ copies of one item, $A$. (This problem is adapted from an example on lifting to get cuts to the knapsack problem from \citealt{BW05}.)

\begin{example}
\label{ex:Five}
$\bm{b_1(5A)= 100};\quad \bm{b_2(5A)= 100};\quad \bm{b_3(5A)= 100};\quad b_4(3A)= 25;\quad b_5(8A)= 90;$
\end{example}

Here, there is clearly a fractional solution, because only 15 of the 17 items sell, and a losing bid will thus activate fractionally in WDP-LP. There is also an easy-to-infer cut, $x_1+x_2+x_3+x_4+x_5 \leq 3$, i.e., ``There can be at most 3 winning bids.'' This cut gives an AWE with $p=[0,90]$, but these fail to be a PME. (If one winner leaves, there is no way to get back to total payments of $270$.)

A second cut can be found, $x_1+x_2+x_3+2x_5 \leq 3$, reflecting, ``Bid 5 knocks out two of the three winning bids.'' Pricing this inferred fact as a second cut gets us PME prices $p=[0,25, 32.5]$. Intuitively, bid 4 sets a minimum value of 25 to be among the winners, and then an additional 65 is required from any combination of two winners so that bid 5 will remain envy-free. One may notice that these correspond to symmetric MRC payments $\rho^{MRC}_i=57.5$, with tight core constraints that say in words, ``Any two winners' payments match the combined bids 4 and 5 which total 115.'' One may also note that $\rho_i^{VCG}=25$ is not in the core for this example.

We observe in Example~\ref{ex:Five} that one may generate an AWE that is not a PME, but then cut further to get a PME.  Further, in Example~\ref{ex:OneDotOne} we observe that a natural WE can be cut to a PME.  From these observations it is natural to ask, what happens more generally if we continue to make cuts past the point of establishing WE prices? To characterize the closure of cut generation for WDP, we next show that any core point can be reached as a feasible AWE by adding at most one cut.

\begin{theorem} 
\label{cuttocore}
     For any core-selection $(\bar{\rho},x^*)$ there exists an AWE with total payments equal to $\bar{\rho}_i$ for each winning bidder $i$.
\end{theorem}

Theorem \ref{cuttocore} ensures that all (personalized) core payments can be reached as a linearization via \emph{some} valid cut. 
%Firstly, for any two cuts $(\alpha^1,\alpha^1_0)$ and $(\alpha^2,\alpha^2_0)$, their sum is a (weaker) valid cut, but we might have preferred the individual cuts for better interpretability. 
%
In Example~\ref{ex:Four}, NWE prices were found to be PME prices, in the core by Proposition~\ref{WEcore}, but distinct from VCG = MRC payments. To arrive at the latter, an intuitive cut $x_1+x_4+x_7\leq 1$ could be added (giving a $k$ index to each bid in the order listed), linearizing the VCG = MRC outcome with $p=[0,0,0,10]$, pricing the added cut as a ``price of 10 to have a two-item bundle.'' This cut was not needed to establish a PME, motivating us to only add cuts as strong as needed for a PME, and not further to MRC, lest we invoke the unrelated goods problem. MRC payments are useful here, however, as minimal PMEs, helping establish that PMEs always exist and are obtainable through valid cuts:

\begin{theorem} \label{towardPME}
     If WE prices $p$ are not a PME, there is a core point with lower total payments. Specifically, $p\bm{A} x^{-\bar{\imath}} < p\bm{A} x^*$ implies $(\bar{\rho},x^*)$ is in the core, where $\bar{\rho}_{i}=pa^{i*}$ for $i \neq \bar{\imath}$ and $\bar{\rho}_{\bar{\imath}}=pa^{\bar{\imath}*}-\epsilon$ for some small $\epsilon>0$.
\end{theorem}

\begin{corollary}\label{MRCisPME}
    If WE prices are MRC, they are also PME prices.
\end{corollary}

If not, Theorem \ref{towardPME} would imply a direction to reduce prices within the core, contradicting the minimum-revenue assumption. Altogether, because MRC prices always exist and Theorem \ref{cuttocore} says they can always be supported as AWE, we have:

\begin{corollary} \label{PMEexists}
     For any CA, PME prices can be generated via cuts.
\end{corollary}

Note that this result is limited to CAs, not extending beyond the assumptions of this paper to combinatorial exchanges (with endogenous supply) as the core and hence MRC is not guaranteed to exist in such a general setting. As an additional observation, using the \citet{AusMil02} characterization of VCG being the unique MRC point when in it is in the core, we find:

\begin{corollary} \label{VCG_PME}
    VCG payments can be linearized to PME prices if and only if VCG $=$ MRC.
\end{corollary}

Also interestingly, despite VCG's uniqueness in MRC whenever in the core, we have the following proposition, proven by the example that follows it.

\begin{proposition} \label{VCGnonU}
    The linearization of VCG in Corollary \ref{VCG_PME} is not necessarily unique.
\end{proposition}

\begin{example}
\label{ex:Six}
$\bm{b_1(4A)= 70};\quad \bm{b_2(4A)= 60}\quad b_3(2A)= 35;\quad b_4(2A)= 20;$ with $c_A=9$ copies of a single item $A$.
\end{example}
Again, excess supply forbids a natural WE for Example~\ref{ex:Six}, but PME prices can be obtained by either the left two cuts or the right two cuts:
\begin{align}
2x_1+2x_2+x_3+x_4 &\leq 4 &&x_1+x_2+x_3 \leq 2\nonumber\\
x_1+x_2+x_3 &\leq 2 &&x_1+x_2+x_4 \leq 2 \nonumber
\end{align}
This induces $p=[0,20,15]$ for left pair or $p=[0,35,20]$ for the right pair, both resulting in linearizations of $\rho^{VCG}=[55,55]$. The left pair can be seen to price pairs of $A$ plus ``only two of bids 1, 2, 3 can win,'' while the right pair says more symmetrically, ``only two of bids 1, 2, $k$ can win'' for $k=3,4$. Notice that the top left cut is the sum of the right two. In our elaboration of MAP prices below, we will find \emph{minimal} cuts that maintain balance among price-setting scenarios, and so we will favor the right two. Final cuts in our procedure belong to the Hilbert basis of the cone of feasible cuts, and so the top left cut could not be selected as an explanation of final payments. (The Hilbert basis is the unique set of minimal integer cuts from which any other integer cut could be generated through non-negative integer combination. See \citet{AARDAL20025}, for example.) 

The following final example was adapted from \citet{GS99} by \citet{D04}, and has a non-PME NWE and artificial PMEs with revenues strictly greater than VCG = MRC:

\begin{example}
\label{ex:Seven}
Using an OXS (OR-of-XOR-of Singletons) bid format for compactness: for each bidder, a bid may be accepted for an item from $\{A,B,C,D\}$ from either column, or both, but using at most one bid per column. For example, using the following bids, $b_1(A)=b_1(C)=8$ and $b_1(AC)=16$, but $b_1(AB)=8$ because only one of the two bid amounts in the first (or any) column can be selected. 
\begin{align}
& \text{Bidder 1} &&\text{\quad Bidder 2} && \text{\quad Bidder 3} \nonumber \\
\nonumber \begin{array}{c}
A \\
B \\
C \\
D \\
\end{array}
&\left| \begin{array}{cc}
\bm{8} & 0 \\
8 & 0 \\
0 & \bm{8} \\
0 & 8 \\
\end{array} \right|
&&\begin{array}{c}
A \\
B \\
C \\
D \\
\end{array}
\left| \begin{array}{cc}
6 & 0 \\
0 & \bm{6} \\
2 & 0 \\
0 & 0 \\
\end{array} \right|
&&\begin{array}{c}
A \\
B \\
C \\
D \\
\end{array}
\left| \begin{array}{cc}
0 & 0 \\
2 & 0 \\
0 & 6 \\
\bm{6} & 0 \\
\end{array} \right|   
\end{align}
\end{example}

One of four possible tied efficient allocations is bolded. This example was originally devised to show that WE prices may exceed VCG payments, even when the gross substitutes property holds. We strengthen this to the more recently studied \emph{strong} substitutes property, which demands gross substitutes even if items are uniquely identified.  Unique NWE prices $p=[6,6,6,6]$ generate 24 in revenue and are not a PME, since the seller could not sell all four items at $p_j=6$ when either bidder 2 or 3 is removed. Payments $\rho^{VCG}=[12,2,2]$ generate only 16 in revenue. VCG is in the core, and thus VCG = MRC. This also shows:
% \begin{figure}
%     \centering
%     \includegraphics[width=0.75\linewidth]{Example7b.png}
%     \caption{The core (in payment space) for Example 7. The pay-as-bid point, minimal NWE point, and VCG are labeled. PME points line in the shaded triangle of the face on $\rho_1=12$, with MAP point $(12,4,4)$ indicated in red.} 
%     \label{fig:GS}
% \end{figure}
% \begin{figure}
%     \centering
%     \includegraphics[width=0.60\linewidth]{Example7C.png}
%     \caption{The core (in payment space) for Example 7. The pay-as-bid point, minimal NWE point, and VCG are labeled. PME points form the shaded triangle on the near face, with the MAP point $(12,4,4)$ indicated in red.}
%     \label{fig:GS}
% \end{figure}

\begin{figure}
    \centering
    \includegraphics[width=0.6\linewidth]{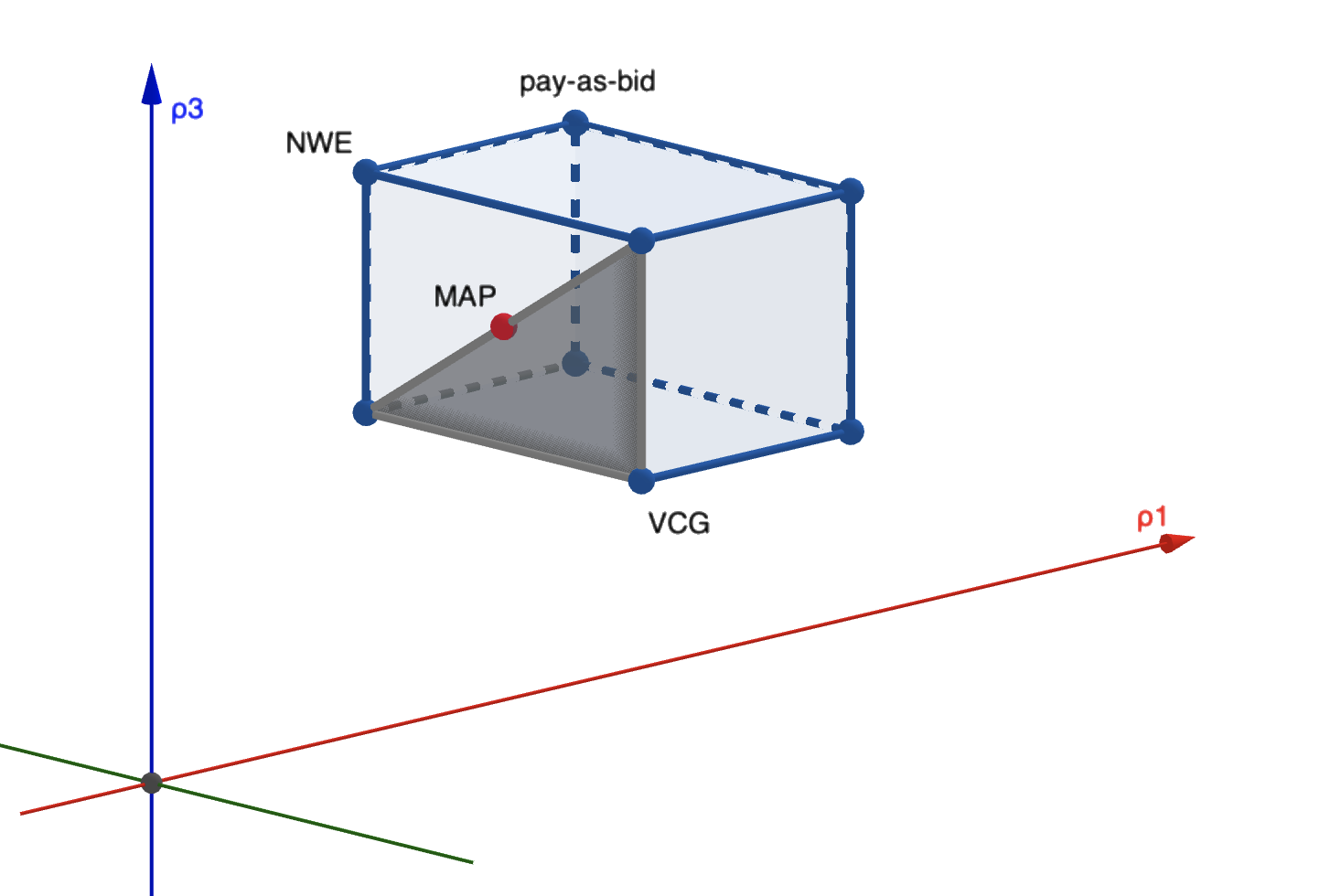}
    \caption{The core (in payment space) for Example~\ref{ex:Seven}. The pay-as-bid, NWE, and VCG  $=$ MRC points are labeled. PME points form the shaded triangle on the near face, with the unique MAP point $\bm{(12,4,4)}$ in red.}
    \Description{}
    \label{fig:GS}
\end{figure}

\begin{corollary} \label{GS}
    The strong substitutes property, while sufficient to guarantee NWE, is not sufficient to guarantee that the NWE is a PME.
\end{corollary}

The MAP outcome for Example~\ref{ex:Seven} can be found after adding two cuts, that say in words, respectively, ``Bidders 1 and 2 cannot both get a most valued bundle,'' and ``Bidders 1 and 3 cannot both get a most valued bundle.'' These cuts yield $p=[4,4,4,4,2,2]$ (the last two prices on artificial items) with total payments of 20 and bidder 1 paying for both artificial items. These payments $\rho_2=\rho_3=4$ lie strictly between the other outcomes for bidders 2 and 3, but with $\rho_1=12$ in all three outcomes. The core of Example~\ref{ex:Seven}  is shown in Figure \ref{fig:GS}, including the unique MAP point in red.

\section{Computing MAP prices}

 Our computational strategy is to find a minimal vector $\phi$ of artificial adjustments for each bid (positive real numbers), which determines final payments. Then in a second procedure, we produce an \emph{elaboration} of these amounts into more meaningful prices, generated as a positive linear combination of minimal cut vectors. The first optimization discovers how much total artificiality \emph{must} be introduced, and the second elaboration operation finds a \emph{simplest} set of cuts that \emph{explains} the artificial amounts.

The first formulation (\ref{mp:albar}) minimizes the total amount of artificial (i.e., not linear in original auction items) adjustments needed to achieve PME:
\begin{align}
    \min\quad & \phi \bm{1} \label{mp:albar} \tag{min-$\phi$}\\
    \text{s.t.\quad} &\phi \bm{D} \geq \bm{0} \label{albar1}\\
    & p\bm{A} + s\bm{B} + \phi \geq b \label{albar2}\\
    &\bm{A}^{-i}x^{-i} \leq c \text{\quad \qquad with }=\text{ when }p_j>0&&\forall i\in \mathcal{I}^*\label{albar3}\\
    &\bm{B}^{-i}x^{-i} \leq \bm{1} &&\forall i\in \mathcal{I}^*\label{albar4}\\
    &\phi x^{-i}  = \phi x^* &&\forall i\in \mathcal{I}^*\label{albar5}\\ 
    &p a^k + \phi_k \leq b_k \text{\quad \qquad when } x_k^{-i}=1 &&\forall i,k\in \mathcal{I}^* \times \mathcal{K}\label{albar6}\\
    & \phi \in \mathbb{R}^K_{\geq 0} \label{albar7}\\
    & s \in \mathbb{R}^I_{\geq 0}\label{albar8}\\
    & p \in \mathbb{R}^J_{\geq 0}\text{\qquad \qquad \quad with }p(c-\bm{A}x^*)=0 \label{albar9}\\
    & x^{-i} \in \{0,1\}^{K-K_i} && \forall i \in \mathcal{I}^* \label{albar10}
\end{align}
 Here we present \ref{mp:albar} compactly, with nonlinear dependencies and quadratic constraints. An equivalent (but longer) linear integer program is presented in Appendix~\ref{FindALI} for completeness.

 If there exists a NWE that is a PME, standard linear item prices are sufficient, and the objective will be zero, with $\phi=\bm{0}$, and thus MAP $=$ NWE. If, on the other extreme, all original items have excess supply, linear prices cannot be used at all, and $\phi$ does all the work of pricing with its nonlinear terms, as occurs in Examples~\ref{ex:Five} and \ref{ex:Six}.

Constraints (\ref{albar1}) like (\ref{cutWEvalid}) ensure a valid cut, while (\ref{albar2}) like (\ref{dual-env}) are WE-defining envy-freeness constraints, but here with the assistance of additional nonlinear $\phi$ terms. Constraints (\ref{albar3}) and (\ref{albar4}) ensure that when any winning bidder $i$ is removed, $x^{-i}$ variables select a feasible WDP solution without $i$, one that sells any item that receives a positive price. Constraints (\ref{albar5}) ensures that (continuous) artificial item $\phi$ sells completely in each $x^{-i}$ scenario. Together with (\ref{albar6}) that ensure IR in each case, (\ref{albar3})-(\ref{albar6}) explicitly fulfill the definition of PME when each bidder is removed. Surplus variables $s$ and thus (\ref{albar8}) can be substituted out because $x^*$ is known, but we leave them in this formulation to show the connection to the simple WDP-dual. The complementary slackness condition shown in (\ref{albar9}) is easy to implement simply as $p_j=0$ if item $j$ has excess supply under $x^*$.

%\todo[inline]{This couple sentences about elaboration feel a little light for something so important.  Can we offer a bit more?  Maybe working the algorithm for Example 6.9, deferring the math to the appendix?  Also this is really where the Hilbert magic happens, so we should talk about it here at least some, and be explicit about how Hilbert relates to elaboration.  The Appendix is long, so we don't want to pull everything in, but maybe we can elaborate on elaboration a bit more than what is here now...}
According to Theorem~\ref{cuttocore} and Corollary~\ref{PMEexists}, formulation~\ref{mp:albar} will always have a solution, given the assumptions of this paper (a single-seller CA normalized to zero reserve prices). If these continuous MAP payment amounts $\phi$ are not all zero, we next perform what we will call \emph{elaboration} to produce MAP item \emph{prices}. (If $\phi=\bm{0}$ no elaboration is needed.)
The elaboration proceeds using a sequence of IPs, each generating a single artificial item (cut), resulting in a set of minimal integral cuts as shown in each of the illustrative examples throughout the paper. Details of this procedure are given in Appendix~\ref{appElab}, which can be summarized as iteratively generating a minimal linear basis for the null space of the tight $\bm{D}$ constraints in formulation~\ref{mp:albar}.

For instance, consider again Example~\ref{ex:Five}. Because only 15 of the 17 copies of natural item $A$ sell, $p_A=0$ (by definition of a WE), and all pricing must be attributed to artificial items. The solution to \ref{mp:albar} is $\phi=(57.5,57.5,57.5,25,90)$, with all tight $\bm{D}$ constraints saying that the $\phi$ amounts of any two winning bids (from $\{1,2,3\})$ must match the losing bids' $\phi$ amounts (from $\{4,5\}$), for example, $\phi_1 + \phi_3\geq\phi_4+\phi_5$. Elaboration looks first for a valid cut with minimal total integer coefficients, breaking ties arbitrarily. Suppose it first finds $x_1+x_2+x_3+x_4+x_5\leq3$, or extracting just LHS coefficients, $\alpha^1=(1,1,1,1,1)$, a Hilbert basis of the cone of tight cuts. The procedure next attempts to generate $\phi$ as a positive multiple, but cannot do so without a positive remainder called $r^t$ at iteration $t$ in the Appendix~\ref{appElab} notation. So far we have $\phi=p^1_1\alpha^1+r^1=25(1,1,1,1,1)+(32.5,32.5,32.5,0,65)$. 
`1Iterating, a next Hilbert basis element gives cut $x_1+x_2+x_3+x_5\leq3$ or $\alpha^2=(1,1,1,0,2)$. The elaboration algorithm can now terminate, as these two basis elements can decompose $\phi$ with zero remainder: $\phi=p^2_1\alpha^1+p^2_2\alpha^2+r^2=25(1,1,1,1,1)+32.5(1,1,1,0,2)+\bm{0}$.

% \section{Computational Experiments on Simulated Data}
% \textbf{Note:} The current draft is being circulated for feedback on the theoretical presentation, though we expect an additional  computational study to provide a numerical evaluation using CATS and SATS data in a future draft. Please do not circulate without the authors' permission.

\section{Conclusions and Future Directions}

We introduce and examine two main ideas for the computation of prices in a general CA:

\begin{itemize}
    \item Non-negative integer valid cuts to WDP have natural interpretations as artificial items, and
    \item PME is a desirable property, with MAP prices as a particular selection of a ``best'' PME.
\end{itemize}

Toward the first idea, we introduce matrix $\bm{D}$ (the directions to optimality from alternate allocations), allowing us to optimize over cuts once a solution to WDP is known. All non-negative integer cuts are in the cone $\alpha\bm{D}\geq\bm{0}$ and thus admit a unique Hilbert basis, a set of irreducible ``facts'' about the relationship of the efficient WDP solution to all other feasible solutions. Pricing these cuts yields the new AWE perspective, and we characterize the closure of cutting algorithms for WDP-LP: any core point can be found as an AWE point by Theorem \ref{cuttocore}.
By Theorem \ref{towardPME} and its corollaries, PME always exist and at one extreme correspond to (i.e., provide a linear decomposition of) MRC payments, a final payment rule which has been used in practice in high-stakes spectrum auctions.

At the other extreme of the PME region, MAP prices emerge as an alternate final-payment rule for any typical CA with the following comparisons to existing payment paradigms:
\begin{enumerate}
    \item MAP prices provide linear prices while typical MRC and VCG implementations provide only lump-sum payments. (Note: MRC outcomes can always be linearly decomposed to PME with additional AWE cuts, but VCG has such a PME decomposition if and only if VCG = MRC.)    
    \item By definition, MAP uses the least amount of artificial price adjustment among PMEs, and so is ``less artificial'' than linearized MRC payments in some cases. 
    \item Ordered by total payments for a fixed set of bids: pay-as-bid $\geq$ NWE $\geq$ MAP  $\geq$ MRC $\geq$ VCG. As intermediate in total-payments, total incentives for misreporting are also intermediate under MAP.
    \item For unit-demand settings min-NWE $=$ MAP $=$ MRC $=$ VCG.
    \item Even under the strong substitutes condition, we may have NWE $>$ MAP $>$ MRC.
    \item MAP payments are additive over combining unrelated-goods markets, unlike MRC.
\end{enumerate}

As a first paper to explore AWE, PME, and MAP prices, we imagine several avenues for further related research.
Given the space limitations of the current submission, details of a larger computational study are deferred to a sequel paper, but we note that MAP requires the solution of several NP-hard WDPs, and so we should expect super-polynomial asymptotic performance unless P = NP. Still, specialized algorithms may prove useful, both for the general problem and for special cases, such as under context-specific preference restrictions.

As with MRC rules \citep[see][]{Bunz22}, a numerical evaluation of the Bayes-Nash equilibria under MAP pricing would be a complex task in itself. Progress on that difficult computational question could be complemented with behavioral lab experiments or other studies of MAP's robustness to behavioral and structural assumptions, such as relaxing quasi-linearity, risk-neutrality, adding budget constraints, etc. Along a different dimension, we have assumed a single, zero-reserve-normalized seller, which allows us to guarantee core outcomes, and thus guarantee PME existence. Results presented here generalize to a combinatorial \emph{exchange} (with multiple buyers and seller) if one is fortunate enough to have a non-empty core, but when the core is empty, PME is not guaranteed to exist. This opens the possibility of approximate MAP prices when the core is empty, but we leave that for future work.

Finally, the WE concept derives from iterative auctions, where price discovery can serve as a common-value aggregation tool. How best to embed MAP pricing into an iterative price-discovery process remains a promising future direction. Here, we introduce MAP as a sealed-bid format, which could be immediately useful on its own or as the final-payment rule (supplementary round) in a combinatorial clock auction. In future work, we plan to study dynamic MAP pricing, in which AWE cuts are generated dynamically as needed during iterative (clock) rounds.

\bigskip

% Acknowledgments here
%\ACKNOWLEDGMENT{The authors gratefully acknowledge the existence of
%the Journal of Irreproducible Results and the support of the Society
%for the Preservation of Inane Research.}

%% TODO: Remove the following when all citations have been included in main text
\newpage
\nocite{*} %% For now, include all references in the bibliography
\bibliographystyle{informs2014}
\bibliography{pricingcuts}

%% Here starts the e-companion (EC)
%%%%%%%%%%%%%%%%%%%%%%%%%%%%%%%%%%%%%%%%%%%%%%%%%%%%%%%%%%

\newpage
\section*{Appendices}
%\appendix
\begin{APPENDICES}

\section{Proofs of Theorems} \label{appProofs}

     Theorem \ref{cuttocore}.
     \textit{For any core-selection $(\bar{\rho},x^*)$ there exists an AWE with total payments equal to $\bar{\rho}_i$ for each winning bidder $i$.}
\medskip

\begin{proof}
     For a proof by contradiction, suppose that the following LP has no solution:  
    \begin{align}
    \min\quad & \phi\bm{0} \nonumber\\
    \text{s.t.\quad} & p a^k + \phi_k = \bar{\rho}_{i(k)} && \forall k \in \mathcal{K}^* \nonumber\\
    &  p a^k + \phi_k \geq b^k - b ^{i(k)*} + \bar{\rho}_{i(k)} 
    && \forall k \in \mathcal{K} \setminus \mathcal{K}^* \nonumber\\
    &\phi \bm{D} \geq \bm{0} \nonumber\\
    & p \in \mathbb{R}^J_{\geq 0}\text{ ,}\quad \phi \in \mathbb{R}^K_{\geq 0} \nonumber
    \end{align}
    By construction, the existence of a solution would be an adjusted set of prices $p$ that equal core payments $\bar{\rho}_i$ for each $i$ while maintaining WE, and are made by a valid cut $(\phi,\phi x^*)$, according to the three sets of linear constraints, respectively. When this formulation is infeasible, its dual must be unbounded, noting that it is feasible when all variables are zero. This dual can be written:
    \begin{align}
    \max\quad & \sum_{k \in \mathcal{K}^*} \bar{\rho}_{i(k)} y_k + \sum_{k \in \mathcal{K} \setminus \mathcal{K}^*} (b^k - b ^{i(k)*} + \bar{\rho}_{i(k)}) z_k
    \nonumber\\
    \text{s.t.\quad} & \sum_{k \in \mathcal{K}^*} a^k_j y_k + \sum_{k \in \mathcal{K} \setminus \mathcal{K}^*} a^k_j z_k \leq 0 && \forall j \in \mathcal{J}\label{T1a}\\
    & \sum_{\ell \in \mathcal{L}} \beta_\ell (1-x^{\ell}_k) \leq -y_k && \forall k \in \mathcal{K}^*\label{T1b}\\
    & z_k \leq \sum_{\ell \in \mathcal{L}} \beta_\ell x^{\ell}_k && \forall k \in \mathcal{K} \setminus \mathcal{K}^*\label{T1c}\\
    & y \in \mathbb{R}^{|\mathcal{K}^*|}_{\leq 0}\text{ ,}\quad z \in \mathbb{R}^{|\mathcal{K} \setminus \mathcal{K}^*|}_{\geq 0}\text{ ,}\quad \beta \in \mathbb{R}^{|\mathcal{L}|}_{\geq 0} \nonumber
    \end{align}
    with $x^{\ell}_k$ terms constant in this formulation from the $\bm{D}$ matrix. Note that $y$ variables are non-positive as a consequence of \eqref{T1b}. Continuous $y$ variables select negative amounts of winning bids, while $z$ variables select positive amounts of non-winning bids. Any increase of a $z$ variable must be offset by at least as much change among $y$ variables in the opposite direction, on an item-by-item basis per constraints \eqref{T1a}. Associated with each $\ell \in \mathcal{L}$ is a maximal feasible solution $x^\ell$; let $\mathcal{K}^\ell=\{k\in\mathcal{K}|x^{\ell}_k=1\}$. Because $\bar{\rho}$ is core-selecting:
    \[
    \sum_{k \in \mathcal{K}^*\setminus \mathcal{K}^\ell} \bar{\rho}_{i(k)} \geq \sum_{k \in \mathcal{K}^\ell \setminus \mathcal{K}^*} (b^k - b ^{i(k)*} + \bar{\rho}_{i(k)})
    \]
    Otherwise, $\ell$ would provide a blocking coalition at $\bar{\rho}$. Multiplying each of these by $\beta_\ell$ and summing over $\mathcal{L}$ produces:
    \[
    \sum_{k \in \mathcal{K}^*} \bar{\rho}_{i(k)} (\sum_{\ell \in \mathcal{L}} \beta_\ell (1-x^{\ell}_k))\geq \sum_{k \in \mathcal{K} \setminus \mathcal{K}^*} (b^k - b ^{i(k)*} + \bar{\rho}_{i(k)}) (\sum_{\ell \in \mathcal{L}} \beta_\ell x^{\ell}_k)  
    \]
    Together with (\ref{T1b}) and (\ref{T1c}), we have:
    \[
    -\sum_{k \in \mathcal{K}^*} \bar{\rho}_{i(k)} y_k \geq \sum_{k \in \mathcal{K} \setminus \mathcal{K}^*} (b^k - b ^{i(k)*} + \bar{\rho}_{i(k)}) z_k
    \]
    But this implies that the maximization objective is always non-positive, contradicting its supposed unboundedness.
\end{proof} 

\bigskip
Theorem~\ref{towardPME}. 
     \textit{If WE prices $p$ are not a PME, there is a core point with lower total payments. Specifically, $p\bm{A} x^{-\bar{\imath}} < p\bm{A} x^*$ implies $(\bar{\rho},x^*)$ is in the core, where $\bar{\rho}_{i}=pa^{i*}$ for $i \neq \bar{\imath}$ and $\bar{\rho}_{\bar{\imath}}=pa^{\bar{\imath}*}-\epsilon$ for some small $\epsilon>0$.}

\medskip
\begin{proof}
    Suppose that for some $\bar{\imath}$ the specified $\bar{\rho}$ is not in the core for any $\epsilon$. Let $\rho$ and $s$ be payments and surplus under $p$.  With $\rho$ in the core by Proposition \ref{WEcore} but $\bar{\rho}$ not in the core after an $\epsilon$ reduction of bidder $\bar{\imath}$'s payment, it must be the case that there is some tight core constraint at $\rho$, indexed by $\mathcal{C}$ with $\bar{\imath} \notin \mathcal{C}$:
    \[
    \sum_{i \in \mathcal{I}^* \setminus \mathcal{C}} \rho_i = w(\mathcal{C}, \mathcal{J}, \mathcal{K^{\mathcal{C}}}) - \sum_{i \in \mathcal{C}} b_{i*}
    \]
    Let each winning bid amount in $w(\mathcal{C}, \mathcal{J}, \mathcal{K^{\mathcal{C}}})$ be given by $b^{\mathcal{C}}_i$ for bundle $a^{\mathcal{C},i}$ for each $i\in\mathcal{C}$, with selection vector $x^{\mathcal{C}}$. Substitute these bid amounts $b^{\mathcal{C}}_i$ into the previous equation for $w(\mathcal{C}, \mathcal{J}, \mathcal{K^{\mathcal{C}}})$ and then substitute $b_{i*}=s_i+\rho_i$, moving $\rho_i$ terms to the LHS to get:
    \[
    \sum_{i \in \mathcal{I}^*}\rho_i=\sum_{i\in\mathcal{C}} (b^{\mathcal{C}}_i - s_i)
    \]
    whose LHS is $p\bm{A}x^*=$ by definition. Next, because $p$ supports a WE, we have $pa^{\mathcal{C},i} \geq b^{\mathcal{C}}_i - s_i$ for each $i \in \mathcal{C}$, summing to yield:
    \[
    p\bm{A}x^* \leq p\bm{A}x^{\mathcal{C}}
    \]
    If we have strict $pa^{\mathcal{C},i} > b^{\mathcal{C}}_i - s_i$ for any $i\in\mathcal{C}$, then this last statement is also strict, contradicting $(x^*,p)$ as a WE. (An efficient $x^*$ must be in the seller's preferred supply set at WE, therefore an alternative $x^{\mathcal{C}}$ with $p\bm{A}x^* < p\bm{A}x^{\mathcal{C}}$ is a contradiction.) If instead none of those envy-freeness constraints are strict, we have each $s_i=b^{\mathcal{C}}_i - pa^{\mathcal{C},i} \geq 0$, and thus $x^{\mathcal{C}}$ is an IR allocation among $\mathcal{I}\setminus \bar{\imath}$ that generates as much price revenue as $x^*$, contradicting the assumption that prices $p$ are not a PME, as specifically witnessed by $\bar{\imath}$.
\end{proof}

\section{Dynamic Generation of \textbf{\textit{D}} } \label{DynD}
Recall from \S~\ref{cutsforWE} that each column of $\bm{D}$ is given by $x^* - x^\ell$ for some alternate set-maximal solution $x^\ell$ of WDP.
For small auctions, such as the numbered illustrative examples of this paper, enumerating all such solutions to WDP is easy (not too expensive computationally). As the size of the auction grows, however, the number of solutions grows exponentially, and it becomes practical to only generate columns of the matrix $\bm{D}$ as needed to guarantee feasibility of an optimization that uses $\bm{D}$. As is standard in constraint generation, this may be accomplished by solving the \emph{separation problem} associated with the region defined by $\bm{D}$. Given a partially generated matrix $\bm{D}^t$ and a potential cut to WDP given by $\alpha x \leq \alpha x^*$ (interpreted here as $\alpha_0 = \alpha x^*$ copies of an artificial item) that satisfies $\alpha \bm{D}^t \geq \bm{0}$, we have the following:

\begin{proposition}
    \textit{Separation at} $\bm{D}^t$: Find solution $x^{\ell +}$ of WDP with bid amounts $b$ replaced by $\alpha$ values. If $\alpha x^{\ell +} \leq \alpha x^*$, then $\alpha x \leq \alpha x^*$ is a valid cut to WDP. Otherwise, append column $x^*-x^{\ell +}$ to $\bm{D}^t$ to form $\bm{D}^{t+1}$ and generate a new potential $\alpha$.
\end{proposition}

\begin{proof}
    The cut $\alpha x \leq \alpha x^*$ is invalid if and only if there is some binary solution $x^\ell$ to WDP with $\alpha x^\ell > \alpha x^*$. The maximality of $x^{\ell +}$ with objective $\alpha$ establishes the first part. For the second part, notice that the appended column yields $\alpha(x^*-x^{\ell +})\geq 0$, or $\alpha x^{\ell +} \leq \alpha x^*$, so that the same $x^{\ell +}$ will not appear again in a dynamic population of $\bm{D}$. Note that this process does not require integer $\alpha$ values, and thus can be used in the process of finding the aggregate continuous cut $\phi$. 
\end{proof}

\medskip
 In practice, commercial IP solvers like CPLEX and Gurobi offer a solution pool feature, listing high-quality feasible integer solutions found during its internal branching. Feasible solutions thus generated may offer a warm-start to the generation of $\bm{D}$. Similarly, WDP solutions found in the course of computing VCG payments or MRC payments can also help to pre-populate $\bm{D}$.

\section{Heuristic for AWE }\label{appA1}
 The focus of the current paper is on PME outcomes, but for completeness we provide a simple iterative process for computing AWE prices (that may not be a PME). It is clear from the related literature that there are often many ways to find cuts to a non-integral IP formulation, so we provide Algorithm~\ref{alg:master} as just one direct cutting algorithm for completeness, leaving the development of faster or more specialized techniques for future research, should one be interested in AWEs that may not be a PME.
 
\begin{algorithm}
    \caption{Heuristic to Find AWE Prices} 
    \label{alg:master} 
    \begin{algorithmic}[0]
    \State Solve WDP and WDP-LP
    \State \textbf{while} $w^f(\mathcal{I}, \mathcal{J}, \mathcal{K})>w(\mathcal{I}, \mathcal{J}, \mathcal{K})$ \textbf{do}
        \State \quad \quad Solve (\ref{cutWE-obj}) with $\phi \in \mathbb{R}^K_{\geq 0}$ and fixed $\phi_0=1$
        \State \quad \quad Find the smallest $\gamma \geq 1$ such that $\forall k \in \mathcal{K}$, $\gamma\phi_k\in \mathbb{Z}$
        \State \quad \quad Update $(\alpha, \alpha_0) \leftarrow (\gamma\phi, \gamma\phi_0)$
        \State \quad \quad Add cut $(\alpha, \alpha_0)$ to all formulations
        \State \quad \quad Re-solve WDP-LP, updating $w^f(\mathcal{I}, \mathcal{J}, \mathcal{K})$ and $x^f$
    %\EndWhile    
    \State Solve WDP-quad-dual (including all cuts as new artificial items)
\end{algorithmic} 
\end{algorithm}

\begin{proposition} 
\label{WEalg}
Algorithm 1 converges to AWE prices for any CA instance.
\end{proposition}

As a sketch of proof, note that a maximally violated cut maximizes $\alpha x^f - \alpha_0$, so that (\ref{cutWE-obj}) finds a maximally violated cut at each iteration, given a fixed $\alpha_0$. Progress therefore must be made at each cut iteration, lowering the WDP-LP objective value, leading to eventual convergence. The scaling integer $\gamma$ is the least common multiple of denominators of reduced rational $\alpha_k$ values, updating them to integer values for a more interpretable artificial item, motivated by our previous examples.

Solved as an LP (i.e., with no integrality constraints within the while loop), Algorithm~\ref{alg:master}'s cut generation appears easy; the difficulty is in generating the large matrix $\bm{D}$, requiring the enumeration of solutions to WDP, which is NP-Hard in general. For small examples, this is not problematic, but as the auction becomes larger, columns of $\bm{D}$ will need to be dynamically generated to manage the size, as in Appendix~\ref{DynD}.

\section{A Linear Mixed-Integer Reformulation of (Min-\texorpdfstring{$\phi$}{a})} \label{FindALI}
To compute MAP prices solve:
\begin{align}
    \min\quad & \phi \bm{1} \label{FindALIL} \tag{Min-$\phi$-L}\\
    \text{s.t.\quad} &\phi \bm{D} \geq \bm{0} \nonumber\\
    & p \Tilde{\bm{A}} +\phi - \phi^* \geq \Tilde{b}\label{FindALI1}\\
    &p \leq M z^T \label{FindALI2}\\
    &\bm{A}^{-i}x^{-i} \leq c &&\forall i\in \mathcal{I}^*\nonumber\\
    &\bm{A}^{-i}x^{-i} \geq c - M(\bm{1}-z)&&\forall i\in \mathcal{I}^*\label{FindALI3}\\
    &\bm{B}^{-i}x^{-i} \leq \bm{1} &&\forall i\in \mathcal{I}^*\nonumber\\
    &\xi^{-i}\bm{1} = \phi x^* &&\forall i\in \mathcal{I}^*\label{FindALI4}\\
    &\xi^{-i} \leq M x^{-i} && \forall i \in \mathcal{I}^*\label{FindALI5}\\
    &\xi^{-i} \leq \phi && \forall i \in \mathcal{I}^*\label{FindALI6}\\
    &\xi^{-i} \geq \phi - M (\bm{1} - x^{-i})&& \forall i \in \mathcal{I}^*\label{FindALI7}\\
    &p \bm{A}^{-i} + \phi \leq b + M(\bm{1} - x^{-i})&&\forall i\in \mathcal{I}^*\label{FindALI8}\\
    & \phi \in \mathbb{R}^K_{\geq 0} \nonumber\\
    & p \in \mathbb{R}^J_{\geq 0}\nonumber\\
    & \xi^{-i} \in \mathbb{R}^K_{\geq 0} && \forall i \in \mathcal{I}^*\nonumber\\
    & x^{-i} \in \{0,1\}^{K-K_i} && \forall i \in \mathcal{I}^*\nonumber\\
    & z \in \{0,1\}^{J} \nonumber
\end{align}
\noindent where $M$ is a sufficiently large positive number. Here, (\ref{FindALI1}) substitutes out surplus variables $s$, resulting in $\Tilde{\bm{A}}$ with each entry given by $\Tilde{\bm{A}}_{j,k}= a^k_j - a^{i(k)*}_j$ , each $\Tilde{b}_k=b_k-b_{i(k)*}$ , and similarly $\phi^*_k = \phi_{i(k)*}$ , in order to treat the cut item the same as a natural items in the substitution. Constraints (\ref{FindALI2}) introduce auxiliary binary variables $z$ that must equal 1 when ever a corresponding price $p$ is positive. Then in each sub-economy $\mathcal{I}^*\setminus i$ , a positively priced item must not have excess supply according to (\ref{FindALI3}).

Next, continuous $\xi^{-i}_k$ variables replace what had been quadratic terms $\phi_k x^{-i}_k$ in (Find $\phi$). Constraints (\ref{FindALI4}) through (\ref{FindALI7}) respectively enforce that: the fictional item sells fully in each sub-economy; that $\xi^{-i}_k$ can only be positive if the bid $k$ is selected in the sub-economy; and that $\xi^{-i}_k=\phi_k$ exactly when $x^{-i}_k=1$. Finally, (\ref{FindALI8}) enforces that any bid selected in any sub-economy is willing to cover the price $pa^k$ with respect to IR (and is loose for any bid not selected).
% \section{<Title of Section B>}
% etc

\section{Elaboration into Interpretable Cuts} \label{appElab}
To do describe the decomposition of a single cut $\phi$, into a set of interpretable integral cuts, we first introduce some additional nomenclature.
Based on the solution of formulation~\ref{mp:albar}, let $\bar{\bm{D}}$ be the tight constraints from (\ref{albar1}), i.e., those columns of $\bm{D}$ such that $\phi\bar{\bm{D}}=\bm{0}$. Call these columns $\mathcal{\bar{L}}$ with $\bar{L}=|\mathcal{\bar{L}}|$. Further, replace any row $k$ of $\bar{\bm{D}}$ with a zero row  wherever $\phi_k=0$. The columns of $\bar{\bm{D}}$ correspond to a set of scenarios in which there is a valid set-maximal allocation including at least one losing bid that exactly matches the total $\phi$ amounts of the winning bids. Thus, $\bar{\bm{D}}$ entries are given by:

  \begin{equation}\nonumber
    \bar{\bm{D}}_{k,\ell} =
    \begin{cases}
        \quad \text{ }1 & \text{if } x^{\ell}_k=0 \text{ , } x^*_k=1 \text{ , and } \phi_k >0\\
        \quad -1& \text{if } x^{\ell}_k=1 \text{ , } x^*_k=0 \text{ , and } \phi_k >0 \\
        \quad \text{ }0 & \text{otherwise} 
    \end{cases}
  \end{equation}

That is, winning bids not in scenario $\ell$ get a positive 1 while non-winning bids in $\ell$ get -1, and the corresponding $\phi$ amounts of these two groups exactly balance for each $\ell$ column. (Winning bids that are in scenario $\ell$ effectively cancel out to a zero entry.) Note that the scenarios captured in the columns of $\bar{\bm{D}}$ are payment-setting to achieve PME by definition. For example, in Example \ref{ex:Five}, where ``any two winners' payments match the combined bids 4 and 5 which total 115,'' the columns in $\bar{\bm{D}}$ correspond to the three (symmetric) scenarios in which one winner could join bids 4 and 5 as matching the total payments of all winners.

Under $\phi\neq\bm{0}$, at least one equation from (\ref{albar5}) is non-zero, and so $\bar{\bm{D}}$ is not empty. From $\phi\bar{\bm{D}}=\bm{0}$ with nonzero $\phi$, $ker(\bar{\bm{D}}^T)=\{\alpha|\alpha \bar{\bm{D}}=\bm{0} \}$ is nontrivial, i.e., $nullity(\bar{\bm{D}}^T)=dim(ker(\bar{\bm{D}}^T))\geq1$, from elementary linear algebra. This $nullity$ will be the dimension of our cut space (i.e., the number of cuts produced).

\begin{definition}
    The set of \emph{critical bids} is $\bar{\mathcal{K}}=\{k\in\mathcal{K}|\exists \ell\text{ with } \bar{\bm{D}}_{k,\ell} \neq 0\}$.
\end{definition}

Critical bids are those appearing in at least one exact payment-matching scenario, and so must (critically) not see a total payment change during elaboration.
Again, assuming an elaboration is needed, this set will be non-empty. The elaboration procedure will generate prices that maintain exactly minimal adjustments $\phi_k$ for all $\bar{\mathcal{K}}$ bids. Notice that $\{k|k\in\mathcal{K}^*\text{ and }\phi_k>0 \} \subset \bar{\mathcal{K}}$ by (\ref{albar5}). 

Note that not all bids $k$ with $\phi_k>0$ are in $\bar{\mathcal{K}}$. That is, there are losing bids for which natural linear prices alone cannot satisfy envy-freeness, but which are not strong enough to help form a payment-matching scenario:

\begin{definition}
    The set of \emph{non-critical bids} is $\mathcal{K}^{nc}=\{k\in\mathcal{K}|\phi_k>0 \text{ and } \forall \ell \text{  } \bar{\bm{D}}_{k,\ell} = 0\}$.
\end{definition}

It is not essential (and we find not desirable for interpretability) for non-critical bids to strictly adhere to their minimal adjustment amounts $\phi_k$, and so the elaboration procedure does not enforce this. (Notice also that there may be a third set of bids, $\mathcal{K} \setminus (\bar{\mathcal{K}} \cup \mathcal{K}^{nc})$, those bids with $\phi_k=0$. These bids satisfy envy-freeness via natural items alone, and may be referred to as \emph{irrelevant} to the elaboration procedure.)

To form an elaboration of $\phi$ into a set of interpretable cuts, name this set of artificial cuts $\mathcal{J}^+$, with $J^+=|\mathcal{J}^+|$. These newly-added artificial items (cuts) $\alpha^j\in \mathcal{J}^+$ extend $\mathcal{J}$ (the natural items) and are generated to be concatenated to $\bm{A}$ as new rows, with prices $p_j$ for $j\in\mathcal{J}^+$ concatenated to $p$.

\begin{definition}
    A set of vectors $\alpha^j$ for $j\in\mathcal{J}^+$ is an \emph{elaboration} of $\phi$ if and only if $\exists p \in \mathbb{R}^{J^+}_{\geq 0}$ with:
    \begin{align}
    \sum_{j \in \mathcal{J}^+} p_j \alpha^j_k &= \phi_k && \forall k \in \bar{\mathcal{K}} \label{aid1} \\
    \sum_{j \in \mathcal{J}^+} p_j \alpha^j_k &\geq \phi_k && \forall k \in \bar{\mathcal{K}}^{nc} \label{aid1.1}\\
    \alpha^j \bm{D} &\geq \bm{0} &&\forall j \in \mathcal{J}^+  \label{aid2}\\
    \alpha^j &\in \mathbb{Z}^K_{\geq 0} &&\forall j \in \mathcal{J}^+ \label{aid3}
\end{align}
\end{definition}

Here, (\ref{aid1}) define a linear decomposition for critical bids, and (\ref{aid1.1}) relax this decomposition for non-critical bids. The direction of this relaxation ensures envy-freeness is maintained; a non-critical bid will remain envy-free with prices totalling more than $\phi_k$. Critical bids, on the other hand, (those in $\bar{\mathcal{K}}$) match total payments of winning bids or are winning, and so strict equality is required.
Constraints (\ref{aid2}) assure that each element of the elaboration is a valid cut to WDP, and (\ref{aid3}) assures that each artificial cut item consists of non-negative integer amounts in each bundle bid $k$, emulating natural-item supply constraints (\ref{suplfeas}).

\begin{proposition}
    \label{prop:null} For any elaboration, $\alpha^j \bar{\bm{D}} = \bm{0}$ for each $j \in \mathcal{J}^+$ with $p_j>0$.
\end{proposition}
\begin{proof}
 Suppose not: there is some column $\bar{\bm{D}}^\ell$ with $\alpha^j\bar{\bm{D}}^\ell>0$ for some $j$ with $p_j>0$. (Notice, $\alpha^j\bar{\bm{D}}^\ell<0$ is impossible by (\ref{aid2})) Next, right-multiply (\ref{aid1}) by $\bar{\bm{D}}$, resulting in $\sum_{j \in \mathcal{J}^+} p_j \alpha^j \bar{\bm{D}}^\ell = 0$ for scenario $\ell$. This is a contradiction because the LHS of this equation has at least one positive term (by supposition) and no negative terms (by (\ref{aid2}) and $p\geq \bm{0}$). 
\end{proof}

\medskip Proposition \ref{prop:null} shows that any cut $\alpha^j$ is necessarily in the kernel (aka null space) of $\bar{\bm{D}}^T$. Intuitively, for any payment-setting scenario to remain payment-setting, changes to the winning bids must match changes to the associated losing bids in each tight scenario. The converse (that a member of the kernel could always be part of some elaboration) is not true in general; first, $\alpha^j_k$ variables with $k\in\mathcal{K}^{nc}$ intersect only zero entries in $\bar{\bm{D}}$ and so a given kernel basis is orthogonal to (\ref{aid1.1}) constraints; and secondly, constraints (\ref{aid2}) may not be satisfied by some elements in the kernel basis. Thus, while any Gaussian elimination of $\bar{\bm{D}}^T$ produces a basis for the kernel, members 
of an arbitrary basis may fail to satisfy (\ref{aid1.1})-(\ref{aid3}) as desired.

We therefore introduce an \emph{elaboration algorithm} with two steps: (i) generate a unique set of basis vectors of $ker(\bar{\bm{D}}^T)$ that satisfy (\ref{aid1}), (\ref{aid2}), and (\ref{aid3}); and then (ii) use an additional optimization to generate a (perhaps non-unique) lifting of the resulting kernel-basis cuts in order to satisfy the non-critical (\ref{aid1.1}).
To accomplish (i) one can iteratively solve a sequence of MIPs and LPs, which we now describe. We find a first kernel-basis cut by solving:
\begin{align}
    \min\quad & \alpha^1 \bm{1} \label{mp:fa1} \tag{Find-$\alpha^{1}$}\\
    &\alpha^{1} \bar{\bm{D}} = \bm{0} \label{fa1.1}\\
    &\alpha^{1} \bm{D} \geq \bm{0} \label{fa1.2}\\
    &\alpha^1 \bm{1} \geq 1 \label{fa1.3}\\
    &\alpha^{1} \in \mathbb{Z}^K_{\geq 0} \label{fa1.4}
\end{align}
Minimizing the sum of the integer $\alpha^1$ terms reflects our preference for economic interpretability. Constraints assure that the cut is a non-zero (\ref{fa1.3}) integer vector (\ref{fa1.4}) in the kernel (\ref{fa1.1}) and also a valid cut (\ref{fa1.2}). 

In order to make progress with each basis cut added, we evaluate how much of $\phi$ can be covered by bases generated thus far by solving an LP, starting at $t=1$:
\begin{align}
    \min\quad & r^t \bm{1} \label{mp:PC} \tag{PriceCuts$^t$}\\
    \sum_{j=1}^{t} p^t_j \alpha^j + r^t &= \phi \label{PC.1} \\
    p^t &\in \mathbb{R}^t_{\geq 0} \label{PC.2}\\
    r^t &\in \mathbb{R}^K_{\geq 0} \label{PC.3}
\end{align}
Here $r$ is a vector of remainders, non-negative amounts of $\phi$ that cannot be covered by price-multiples of the $t$ cut vectors generated thus far. When at some iteration $t$ we arrive at zero objective (i.e., $r=\bm{0}$), part (i) of elaboration can terminate, having successfully found  basis vectors of $ker(\bar{\bm{D}}^T)$ that satisfy (\ref{aid1}), (\ref{aid2}), and (\ref{aid3}). 
For example, if $nullity(\bar{\bm{D}}^T)=1$, as the lone basis vector of $ker(\bar{\bm{D}}^T)$, $\alpha^1$ would generate any kernel element as a multiple, and thus the optimal solution to \ref{mp:PC} is zero at the first iteration. If not, it is not time to terminate, and so we increment $t$ and generate the next cut vector by solving the following MIP (\ref{mp:fat}) for $t\geq 2$:
\begin{align}
    \min\quad & \alpha^t \bm{1} \label{mp:fat} \tag{Find-$\alpha^{t}$}\\
    &\alpha^{t} + \sum_{j=1}^{t-1} \gamma_j \alpha^j = \delta^t \label{fat1} \\
    &\alpha^{t} \bar{\bm{D}} = \bm{0} \label{fat2}\\
    &\alpha^{t} \bm{D} \geq \bm{0} \label{fat3}\\
    &\delta_k \geq 0 && \quad\quad\forall k\in\bar{\mathcal{K}} \text{ }| \text{ }r^{t-1}_k \geq 0\label{fat4}\\
    &\delta_{\Tilde{k}} \geq \epsilon &&\text{for some }\Tilde{k}\in\bar{\mathcal{K}} \text{ with }r^{t-1}_k > 0\label{fat5}\\
    &\delta_k = 0 && \quad\quad\forall k\in\bar{\mathcal{K}}\text{ }| \text{ }r^{t-1}_k = 0\label{fat6}\\
    &\alpha^{t} \in \mathbb{Z}^K_{\geq 0} \label{fat7}\\
    & \gamma \in \mathbb{R}^{t-1}\label{fat8}\\
    & \delta^t \in \mathbb{R}^K_{\geq 0} \label{fat9}
\end{align}
As in \ref{mp:fa1}, each $\alpha^t$ will be an integer vector in the kernel and also a valid cut. Constraints (\ref{fat4}) and (\ref{fat5}) assure that non-negative progress can be made in reducing positive $r^{t-1}$ terms, with strict progress for an arbitrarily chosen $\Tilde{k}$, while not disrupting $r^{t-1}$ terms that have already reached zero. This process of solving \ref{mp:fat} then \ref{mp:PC} continues until $r^t=\bm{0}$, at most $t=nullity(\bar{\bm{D}}^T)$ iterations.
\begin{proposition}
    After $J^+=nullity(\bar{\bm{D}}^T)$ iterations, (\ref{mp:fat}, \ref{mp:PC}) generates artificial-item cuts that satisfy elaboration conditions (\ref{aid1}), (\ref{aid2}), and (\ref{aid3}), and thus form an undominated integer basis of $ker(\bar{\bm{D}}^T)$.
\end{proposition}
\begin{proof}
    First, \ref{mp:fat} is feasible for $t\in\{1,...J^+\}=\mathcal{J}^+$. The feasible set $\{\alpha \in \mathbb{Z}^K_{\geq 0}| \alpha \bar{\bm{D}}=\bm{0}, \alpha \bm{D} \geq \bm{0}\}$ consists of the integer lattice points of a pointed polyhedral cone, and thus is finitely generated by a unique Hilbert basis \citep[see][]{AARDAL20025}. This basis is a collection of component-wise minimal row vectors, together forming a matrix $H$ generating the same feasible set as 
    $\{\alpha | \alpha = p H, p \in \mathbb{Z}^{|H|}_{\geq0}\}$, which can also be characterized as a Graver basis \citep[see][]{le2024equivariant}. The theoretical existence of this Hilbert basis guarantees integer solutions to each \ref{mp:fat}. Because the Hilbert basis definition demands \emph{integer} weights, $|H|>nullity(\bar{\bm{D}})$ is possible. Here, these weights are prices and thus continuous, and so a standard linear basis of size $nullity(\bar{\bm{D}})$ is appropriate. As a consequence, the generated $\alpha^t$ cuts are potentially a strict subset of $H$: generators of the standard linear cone generated by $H$. Note that there is no loss in restricting from the integer cone $\{\alpha \in \mathbb{Z}^K_{\geq 0}| \alpha \bar{\bm{D}}=\bm{0}, \alpha \bm{D} \geq \bm{0}\}$ to the linear cone of $H$, as $\phi$ belongs to the latter, as it is rational given rational coefficients to the original problem, and therefore an integer multiple belongs to the former. 

    To see that progress in generating bases is made at each $t$, multiply (\ref{fat1}) by an arbitrary $p_t>0$ and add it to the $t-1$ indexed (\ref{PC.1}), yielding:
    \[
    p_t\alpha^{t} + \sum_{j=1}^{t-1} (p_t\gamma_j+p^{t-1}_j) \alpha^j + r^{t-1} -p_t\delta^t = \phi\text{ .}
    \]
    This new $\alpha^t$ extends the basis to improve the objective of \ref{mp:PC}, because $r^{t-1} -p_t\delta^t \leq r^{t-1}$ by construction, with strict inequality for the chosen $\Tilde{k}$. This improvement would not be possible if $\alpha^t$ was a linear combination of previous $\alpha^j$s (or else optimality at $t-1$ is contradicted), thus the basis must grow. Note that this only shows that the new basis establishes an improving direction when $p_t$ is small enough that $p_t\gamma_j+p^{t-1}_j>0$ for $j<t$, and so we still solve \ref{mp:PC} in between each \ref{mp:fat} iteration. 
\end{proof}

\medskip The solution set of the series \ref{mp:fat} can be solved by an IP solver as written, or be enumerated as in Algorithm \ref{alg:enum}, presented in Appendix \ref{Findfat}, which uses a depth-first backtracking search tree. That intuitive characterization can be implemented ``by hand'' on small instances. 

For elaboration part (ii), we need $\alpha^j_k$ for $(j,k) \in \mathcal{J}^+\times \mathcal{K}^{nc}$ that together satisfy constraints (\ref{aid1.1}). Strengthening the already generated (part (i) elaborated) cuts to include \emph{non-critical} bids is accomplished by solving the following IP, a kind of \emph{group lifting} of these variables:
\begin{align}
    \min\quad  \sum_{j\in\mathcal{J}^+} &\sum_{k\in \mathcal{K}^{nc}}\alpha^j_k \label{mp:ncl} \tag{$nc$-lift}\\  
    \sum_{j \in \mathcal{J}^+}& p_j \alpha^j_k \geq \phi_k && \forall k \in \mathcal{K}^{nc} \label{ncl1}\\
    \sum_{k\in \mathcal{K}^{nc}}& x^{\ell}_k \alpha^j_k \leq \sum_{k\in \mathcal{K}^* } x^{\ell}_k \alpha^j_k - \sum_{k\in \mathcal{K}} x^{\ell}_k \alpha^j_k\label{ncl2} && \forall (j,\ell)\in\mathcal{J}^+\times \mathcal{L}^{nc}\\
    & \alpha^j_k \in \mathbb{Z}_{\geq 0} && \forall (j,k) \in \mathcal{J}^+ \times \mathcal{K}^{nc} \nonumber
\end{align}
Notice that only $\alpha^j_k$ for $(j,k) \in \mathcal{J}^+ \times \mathcal{K}^{nc}$ are decision variables in this formulation, with all $p_j$, $x$, and other $\alpha^j_k$ terms fixed from part (i) of the elaboration. Constraints (\ref{ncl1}) repeat (\ref{aid1.1}), while (\ref{ncl2}) ensure that each artificial item $j \in \mathcal{J}^+$ remains a valid cut for WDP after lifting in $nc$ variables (equivalent to $\alpha \bm{D}\geq \bm{0}$ but focusing the decision variables to the LHS of this formulation and using only $\mathcal{L}^{nc}$ columns of $\bm{D}$).

To see that this this IP must be feasible, observe that we could selectively increase both bid amount $b^k$ and $\phi_k$ amounts by the same level for some bids $k\in\mathcal{K}^{nc}$ to result in a new auction in which all such $k$ are now critical. In that new auction, elaboration part (i) would give us cuts including $\mathcal{K}^{nc}$, and thus a feasible solution to \ref{mp:ncl}. Still, such an adjusted auction approach would tend to overstate the prices quoted to $nc$ bids and distort the interpretability of artificial items. We thus pursue the current formulation to find more informative prices.

For the concise examples presented here, we tend not to include extraneous bids that would end up being non-critical. An exception is Example \ref{ex:Three}: bid $b_6$ is non-critical because it offers only 34 for all items, when the winning and price setting outcomes are each able to produce at least 35. The elaboration procedure (i) would first generate the cut $x_1+x_2+x_3+x_4+x_5 \leq 2$ among winning and critical bids, and then (ii) would solve \ref{mp:ncl} to lift in $b_6$ with a coefficient of 2: relative to any other set-maximal allocation, $b_6$ knocks out both bids that would pay for one copy of the artificial item. Similarly, if one were to (say) introduce an additional bid $b_6(5A)=40$ into Example \ref{ex:Five}, this new bid would be non-critical and lifted in to have coefficients of 1 and 1 for the generated artificial items. This $nc$ bid would face a payment of 57.5, which is strictly non-envied and symmetric to the winning bids, which bid on the same set.

As one attempts larger and thus computationally harder auction instances, the treatment of $\mathcal{K}^{nc}$ in elaboration part (ii) could be relaxed or eliminated altogether without abandoning the main results presented. A few choices for elaboration part (ii) are immediately possible:
\begin{itemize}
    \item \textbf{Minimal:} Implement elaboration exactly as stated here. For a selected final outcome based on a particular $\phi$. Non-critical bids are guaranteed to be among the AWE that minimize the total artificial price quotes to $nc$ bids.
    \item \textbf{Integer Feasible:} Implement elaboration as stated here, but stopping at any feasible solution to \ref{mp:ncl}, perhaps when a computational time-limit or optimality-limit is reached. The $nc$ bids see valid linear price quotes that leave them envy-free, but these may not be minimal among such prices, a slightly worse but largely agreeable outcome.
    \item $nc$\textbf{-Nonlinear Only:} Implement elaboration without running part (ii). Final prices are totally explained by the relationship to critical bids, including those that form the price-match allocations. Critical bids see the linear decomposition of $\phi_k=\sum_{j \in \mathcal{J}^+} p_j \alpha^j_k$, as always, while $nc$ bids would only see a lump-sum non-linear payment component $\phi_k$ (above linear prices on original items). This lump-sum would cover envy-freeness, and the optimality of \ref{mp:albar} would still guarantee that the $nc$ bid could not be part of any allocation that would match total payments of the winners.
\end{itemize}

\section{A Search Heuristic for Enumerating Solutions to \texorpdfstring{\ref{mp:fat}}{Find-alpha\^t}} \label{Findfat}
Starting with a single starting node, when reaching the Select function, add tree branches from the current node to one new node for each feasible selection $k$, and proceed forward to a node by selecting an arbitrary $k$. When a Backtrack command is encountered, run the code backward to the most recent Select command. If there is an unexplored $k$ at that selection, take one such arbitrary $k$ instead of one previously used, and run the code forward with this alternate selection. If there is no unexplored node (because all have been explored on a previous walk forward) then continue backward through the while loop, subtracting previous $e_k$ updates to $\alpha$, until encountering a Select command with unexplored successors. Proceed forward to explore.

 \begin{algorithm}
    \caption{Enumerate feasible solutions to \ref{mp:fat} for $t=1$ to $t=nullity(\bar{\bm{D}}^T)$.}
    \label{alg:enum} 
    \begin{algorithmic}[0]
    \State Initialize: $keep=\emptyset$; $out=\emptyset$; $\alpha=\bm{0}$
    \State Select an arbitrary $k \in \bar{\mathcal{K}} \cap \mathcal{K}^*$, and set $\alpha_k=1$.
    
    \State \textbf{while} $\alpha \bar{\bm{D}} \neq \bm{0}$ \textbf{do}
        \State \qquad Update $\mathcal{K}_{near}=\text{argmin}_{k\in \bar{\mathcal{K}}|(\alpha+e_k)\notin keep\cup out} \text{ }\lVert  (\alpha + e_k)\bar{\bm{D}} \rVert_1$
        \State \qquad Update $\mathcal{K}_{best}=\text{argmin}_{k\in \mathcal{K}_{near}} \text{ } \alpha_k$
        \State \qquad Select an arbitrary $k\in \mathcal{K}_{best}$ and set $\alpha\gets\alpha +e_k$

    \State \textbf{if} $\alpha \bm{D}\geq\bm{0}$ \textbf{then}
                \State \qquad Add $\alpha$ to $keep$
    \State \textbf{else}
                \State \qquad Add $\alpha$ to $out$

    \State \textbf{if} $|keep|=nullity$ \textbf{then}
                \State \qquad Terminate
    \State \textbf{else} 
                \State \qquad Backtrack

\end{algorithmic} 
\end{algorithm}
This algorithm is easy to implement ``by hand'' for small problems and adds intuition for the ease of generating kernel-basis cuts.

\end{APPENDICES}

%%%%%%%%%%%%%%%%%
\end{document}